\documentclass[11pt]{article}
\usepackage{geometry}
\usepackage{amsmath, amssymb, amsthm}
\usepackage{graphicx}
\usepackage{bm}
\usepackage[linesnumbered,ruled]{algorithm2e}
\SetKwRepeat{Do}{do}{while}
\usepackage{fullpage}
\newtheorem{theorem}{Theorem}
\newtheorem{lemma}{Lemma}

\newtheorem{claim}{Claim}

\title{The Strongly Stable Matching Problem with Closures}
\author{Naoyuki Kamiyama%
\thanks{This work was supported by JSPS KAKENHI Grant Number JP20H05795 and 
JST ERATO Grant Number JPMJER2301, Japan.}
}
\date{\small Institute of Mathematics for Industry, Kyushu University, Fukuoka, Japan\\
{\ttfamily kamiyama@imi.kyushu-u.ac.jp}}

\begin{document}

\maketitle

\begin{abstract}
In this paper, we consider one-to-one matchings between 
two disjoint groups of agents. 
Each agent has a preference over a subset of the agents in the other group, and 
these preferences may contain ties. 
Strong stability is one of the stability concepts in 
this setting. 
In this paper, we consider the following variant 
of strong stability, and we consider 
computational complexity issues for this solution concept. 
In our setting, 
we are given a subset of the agents on one side.
Then 
when an agent in this subset is not 
matched to any partner, 
this agent cannot become a part of a blocking pair. 
In this paper, we first prove that the problem of determining 
the existence of a stable matching in this setting 
is NP-complete.
Then we give two polynomial-time solvable cases of our problem. 
Interestingly, one of our positive results 
gives a unified approach to the strongly stable matching 
problem and the envy-free matching problem. 
\end{abstract} 

\section{Introduction}

In this paper, we consider one-to-one matchings between 
two disjoint groups of agents. 
Each agent has a preference over a subset of the agents in the other group. 
The most famous solution concept in this setting 
is a stable matching~\cite{GaleS62}. 
In the basic setting of the stable matching, 
each agent has a strict preference, i.e., a preference 
of each agent does not contain ties. 
By contrast, we consider the situation where a preference of an agent 
may contain ties, i.e., an agent may be indifferent between 
potential partners. 
It is known that
ties in preferences make a big difference
in 
the stable matching problem
(see, e.g., \cite{IwamaM08} and \cite[Chapter~3]{Manlove13}). 

Under preferences with ties, the following three 
stability concepts have been mainly considered. 
The first solution concept is weak stability. 
This property guarantees that 
there does not exist an unmatched pair of agents $a,b$ such that 
$a$ (resp.\ $b$) prefers $b$ (resp.\ $a$) to the current partner.
Irving~\cite{Irving94} proved 
that there always exists a weakly stable matching, and 
a weakly stable matching can be found 
in polynomial time by slightly modifying the algorithm 
of Gale and Shapley~\cite{GaleS62}. 
The second solution concept is super-stability.
This property guarantees
that there does not exist an unmatched pair of agents $a,b$ 
such that $a$ (resp.\ $b$) prefers $b$ (resp.\ $a$) to the current partner, or is 
indifferent between $b$ (resp.\ $a$) and the current partner.
The last solution concept is strong stability,
which 
is the main topic of this paper. 
This property guarantees that 
there does not exist an unmatched pair of agents $a,b$ such that
(i) $a$ prefers $b$ to 
the current partner, and (ii) $b$ 
prefers $a$ to the current partner, or is indifferent between 
$a$ and the current partner. 
It is known that 
a super-stable matching and a strongly stable matching 
may not exist (see \cite{Irving94}). 
Thus, one of the central algorithmic problems 
for these two solution concepts
is the problem of determining the existence of 
a matching satisfying these stability concepts. 

In this paper, 
we introduce a variant of the strongly stable matching 
problem, 
which we call 
the strongly stable matching problem with closures, and 
we consider 
computational complexity issues for this problem. 
In order to explain our setting, 
we consider a matching problem between 
doctors and hospitals. 
In our setting, we are given a subset $S$ of 
the hospitals.
Recall that, in the classical setting, 
an unmatched pair of a doctor and a hospital 
has an incentive to make a new pair between them 
when each involved agent prefers the new partner to the current partner 
or is unmatched in the current matching. 
On the other hand, 
in our setting, 
if the hospital is in the set $S$, then any doctor cannot 
claim the empty position in this hospital. 
That is, 
for each hospital $h \in S$, if $h$ is not matched to 
any doctor, then $h$ is closed and 
becomes unavailable for any doctor. 
(See Section~\ref{section:formulation} for the formal
definition of our stability.) 
This closure operation 
is inspired by 
the closure operation in assignment problems with lower quotas
(e.g., \cite{Kamiyama13,MonteT13}). 

Our contribution is summarized as follows. 
First, we prove that 
the problem of determining 
the existence of a stable matching in this setting 
is NP-complete even in a very restricted setting (see Theorem~\ref{theorem:hardness}). 
Then we consider special cases of our problem that 
can be solved in polynomial time. 
We first consider a special case where, in a sense, the preferences of doctors 
are separated 
(see Section~\ref{section:separated} for the formal definition of 
this special case). 
We prove that, in this setting, our problem can be solved in polynomial time 
(see Theorem~\ref{theorem:separate}). 
Interestingly, the result in this special case 
gives a unified approach to the strongly stable matching 
problem~\cite{Irving94,Manlove99} and the envy-free matching problem~\cite{GanSV19} (see Section~\ref{section:envy}).
Next, we consider the case where 
each doctor accepts at most two hospitals, 
and 
we prove that this special case can be solved 
in polynomial time
(see Section~\ref{section:bounded_degree}). 

\paragraph{Related work.} 

Here we summarize related work on strongly stable matchings. 

Irving~\cite{Irving94} 
introduced the solution concept of strong stability, and
gave  
a polynomial-time algorithm for 
the strongly stable matching problem 
(see also \cite{Manlove99}). 
Furthermore, 
Kunysz, Paluch, and 
Ghosal~\cite{KunyszPG16}
considered 
characterization of the set of all strongly stable matchings. 
Kunysz~\cite{Kunysz18} considered 
the weighted version of the strongly stable matching
problem. 

In the many-to-one setting, 
Irving, Manlove, and Scott~\cite{IrvingMS03} and 
Kavitha, Mehlhorn, Michail, and Paluch~\cite{KavithaMMP07}
proposed polynomial-time algorithms for the 
strongly stable matching
problem. 

In the many-to-many setting, 
the papers~\cite{ChenG10,Kunysz19,Malhotra04} considered 
the strongly stable matching problem.
Furthermore, Olaosebikan and Manlove~\cite{OlaosebikanM20}
considered strong stability in the 
student-project allocation problem. 

For the situation where a master list is given, 
Irving, Manlove, and Scott~\cite{IrvingMS08} gave 
a simple polynomial-time algorithm for  
the strongly stable matching problem. 
Furthermore, O'Malley~\cite{OMalley07} gave 
a polynomial-time algorithm for 
the strongly stable matching problem
in the many-to-one setting.
Under matroid constraints, 
Kamiyama~\cite{Kamiyama15,Kamiyama19} gave 
polynomial-time algorithms for 
the strongly stable matching problem
in the many-to-one and many-to-many settings. 

\section{Problem Formulation} 
\label{section:formulation} 

For each positive integer $z$, 
we define $[z] := \{1,2,\dots,z\}$. 

In this paper, a finite simple undirected bipartite graph 
$G$ is denoted by $G = (D, H; E)$, and 
we assume that the vertex set of $G$ is partitioned into 
$D$ and $H$, and 
every edge in $E$ connects a vertex in $D$ and 
a vertex in $H$.
In this paper, 
we call a vertex in $D$ (resp.\ $H$) a \emph{doctor} (resp.\ \emph{hospital}). 
For each doctor 
$d \in D$ and each hospital $h \in H$, 
if there exists an edge in $E$ between $d$ and $h$, then 
this edge is denoted by $(d,h)$. 
For 
each subset $F \subseteq E$
and 
each subset 
$X \subseteq D$ (resp.\ $X \subseteq H$), 
we define $F(X)$ 
as the set of edges $(d,h) \in F$
such that
$d \in X$ (resp.\ $h \in X$). 
For 
each subset $F \subseteq E$
and 
each vertex $v \in D \cup H$, 
we write $F(v)$ instead of $F(\{v\})$. 
For each subset $F \subseteq E$ and 
each subset $X \subseteq D$, 
we define $\Gamma_F(X)$ as the set of 
hospitals $h \in H$ 
such that $F(X) \cap E(h) \neq \emptyset$. 
For 
each subset $F \subseteq E$
and each doctor $d \in D$, 
we write $\Gamma_F(d)$ instead of $\Gamma_F(\{d\})$. 
For each subset $F \subseteq E$, we define 
$D[F]$ as the set of doctors $d \in D$ such that
$F(d) \neq \emptyset$. 
For each subset $X \subseteq D \cup H$, we define 
$G \langle X \rangle$ as the subgraph of $G$ induced by 
$X$. 

The \emph{strongly stable matching problem with closures}, which is 
the main topic of this paper, is defined as follows. 
In this problem, we are given a finite simple undirected bipartite graph 
$G = (D, H; E)$. 
Furthermore, we are given a subset $S \subseteq H$.
For each vertex 
$v \in D \cup H$, we are given 
a transitive binary relation 
$\succsim_v$ on $E(v) \cup \{\emptyset\}$
satisfying the following conditions. 
\begin{itemize}
\item
For every pair of elements 
$e,f \in E(v) \cup \{\emptyset\}$, at least one of $e \succsim_v f$ 
and $f \succsim_v e$ holds.
\item
For every edge $e \in E(v)$, we have $e \succsim_v \emptyset$ and 
$\emptyset \not\succsim_v e$.
\end{itemize}
For each vertex $v \in D \cup H$
and each pair of 
elements $e,f \in E(v) \cup \{\emptyset\}$, 
if $e \succsim_v f$ and $f \not\succsim_v e$ 
(resp.\ $e \succsim_v f$ and $f \succsim_v e$),
then we write 
$e \succ_v f$ 
(resp.\ $e \sim_v f$).  

For each subset $\mu \subseteq E$, $\mu$ is called a \emph{matching in $G$} if
$|\mu(v)| \le 1$ for every vertex $v \in D \cup H$. 
For each matching $\mu$ in $G$ and each vertex $v \in D \cup H$ such that 
$\mu(v) \neq \emptyset$, 
we do not distinguish between $\mu(v)$ and the unique element in $\mu(v)$. 

Let $\mu$ be a matching in $G$.
For each edge
$e = (d,h) \in E \setminus \mu$, 
we say that \emph{$e$ weakly} (resp.\ \emph{strongly}) \emph{blocks $\mu$ on $d$}
if $e \succsim_d \mu(d)$ (resp.\ 
$e \succ_d \mu(d)$).
For each edge
$e = (d,h) \in E \setminus \mu$, 
we say that \emph{$e$ weakly} (resp.\ \emph{strongly}) \emph{blocks $\mu$ on $h$}
if the following conditions are satisfied. 
\begin{itemize}
\item
If $h \in S$, then 
$\mu(h) \neq \emptyset$ and 
$e \succsim_h \mu(h)$ (resp.\ $e \succ_h \mu(h)$). 
\item 
If $h \not\in S$, then 
$e \succsim_h \mu(h)$ (resp.\ $e \succ_h \mu(h)$). 
\end{itemize}
For each edge $e = (d,h) \in E \setminus \mu$, 
we say that \emph{$e$ blocks $\mu$} if 
the following conditions are satisfied. 
\begin{itemize}
\item
$e$ weakly blocks $\mu$ on both $d$ and $h$. 
\item
$e$ strongly blocks $\mu$ on at least one of $d$ and $h$.
\end{itemize}
Then a matching $\mu$ is said to be \emph{stable} if
any edge $e \in E \setminus \mu$
does not block $\mu$.  

It is known that 
a stable matching does not necessarily exist
even when $S = \emptyset$~\cite{Irving94}. 
Thus, the goal of 
the strongly stable matching problem with closures
is to determine
whether there exists a stable matching in $G$ and to find one 
if exists. 

Here we explain how we describe the preference of a vertex.
Let $d$ be a doctor in $D$, and we assume that 
\begin{equation*}
E(d) = \{(d,h_1), (d,h_2), (d,h_3)\}, \ \ \  
(d,h_1) \succ_d (d,h_2) \sim_d (d,h_3).
\end{equation*}
In this case, we describe the preference of $d$ as 
$d \colon h_1 \succ h_2 \sim h_3$.
We describe the preference of a hospital in the same way. 

\section{Hardness} 

In this section, we prove the following hardness result. 
We define the \emph{decision version} of the strongly stable matching problem with closures 
as the problem of determining 
whether there exists a stable matching in $G$. 

\begin{theorem} \label{theorem:hardness}
The decision version of the strongly stable matching problem with closures is ${\rm NP}$-complete 
even when the following conditions are satisfied. 
\begin{itemize}
\item
For every doctor $d \in D$ and 
every pair of edges $(d,h),(d,h^{\prime}) \in E(d)$ such that 
$h \in S$ and $h^{\prime} \notin S$, 
we have $(d,h) \succ_d (d,h^{\prime})$.  
\item
For every vertex $v \in D \cup H$, we have 
$|E(v)| \le 3$. 
\end{itemize}
\end{theorem}
\begin{proof}
Since we can determine whether a given matching in $G$ is stable in polynomial time, 
the decision version of 
the strongly stable matching problem with closures
is in NP. 

We prove the NP-hardness of the decision version of 
the strongly stable matching problem with closures 
by reduction from 
$(3, {\rm B}2)$-{\sc sat} defined as follows. 
In $(3, {\rm B}2)$-{\sc sat}, we are given a set ${\sf V}$ of 
$n$ variables 
$\alpha_1,\alpha_2,\ldots,\alpha_n$ and 
$m$ clauses $C_1,C_2,\ldots,C_m$. 
Define $\neg {\sf V}$ as the set of negations 
$\neg \alpha_1, \neg\alpha_2,\ldots,\neg\alpha_n$ 
of the variables in ${\sf V}$.
Each clause is a subset of 
${\sf V} \cup \neg {\sf V}$ that contains exactly three elements.
Furthermore, 
for every variable $\alpha$, 
the number of clauses containing $\alpha$ is exactly two, and 
the number of clauses containing $\neg \alpha$ is exactly two.
Then the goal of $(3, {\rm B}2)$-{\sc sat} is to determine whether 
there exists a function $\phi \colon {\sf V} \to \{0,1\}$ such that, 
for every integer $t \in [m]$, at least one element in $C_t$ 
is $1$ under $\phi$, where we define the 
negation of $1$ (resp.\ $0$) as 
$0$ (resp.\ $1$).
It is known that 
$(3, {\rm B}2)$-{\sc sat} is {\rm NP}-complete~\cite[Theorem~1]{BermanKS03}. 

Assume that we are given an instance of 
$(3, {\rm B}2)$-{\sc sat}. 
For each integer $t \in [m]$, we arbitrarily arrange 
the elements in $C_t$ and 
we define $\ell_{t,j}$ as the $j$th element in $C_t$.
Then we construct an instance of the strongly 
stable matching problem with closures as follows. 

\begin{figure*}[t]
\begin{minipage}{0.25\hsize}
\centering
\includegraphics[width=3cm]{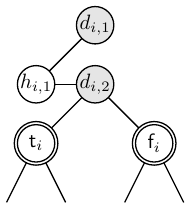}
\par(a)
\end{minipage}
\begin{minipage}{0.3\hsize}
\centering
\includegraphics[width=5cm]{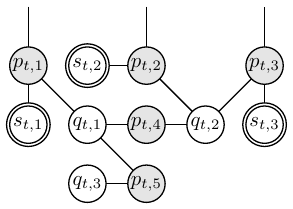}
\par(b)
\end{minipage}
\begin{minipage}{0.43\hsize}
\centering
\includegraphics[width=5.5cm]{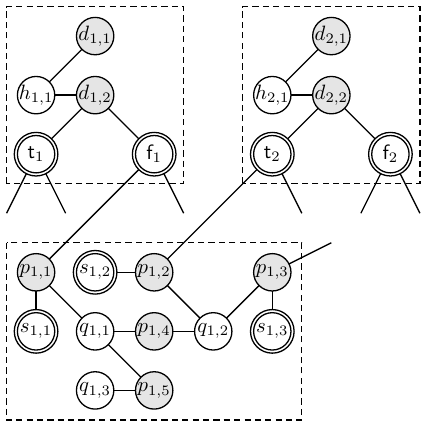}
\par(c)
\end{minipage}
\caption{(a) The gadget for the variable $\alpha_i$.
(b) The gadget for the clause $C_t$.
(c) $\ell_{1,1} = \neg \alpha_1$ and $\ell_{1,2} = \alpha_2$.}
\label{fig:gadget}
\end{figure*}

For each variable $\alpha_i$, we prepare the gadget 
described in Figure~\ref{fig:gadget}(a). 
In this paper, gray (resp.\ white) vertices represent doctors 
in $D$ (resp.\ hospitals in $H$).
Furthermore, the vertices described by double circles are 
hospitals in $S$. 
That is, in the gadget in Figure~\ref{fig:gadget}(a), 
${\sf t}_i, {\sf f}_i \in S$ and 
$h_{i,1} \notin S$.
The hospitals ${\sf t}_i, {\sf f}_i$ are connected to 
some vertices in the gadgets corresponding to the clauses containing 
literals of $\alpha_i$.  
The preferences of $h_{i,1}, d_{i,2}$ are defined as 
\begin{equation*}
h_{i,1} \colon d_{i,1} \sim d_{i,2}, \ \ \  
d_{i,2} \colon {\sf t}_i \sim {\sf f}_i \succ h_{i,1}.
\end{equation*}

For each clause $C_t$, we prepare the gadget 
illustrated in Figure~\ref{fig:gadget}(b). 
That is, $s_{t,1}, s_{t,2}, s_{t,3} \in S$ and 
$q_{t,1}, q_{t,2}, q_{t,3} \notin S$.
The doctors $p_{t,1}, p_{t,2}, p_{t,3}$ are connected to 
some vertices in the gadgets corresponding to variables in $C_t$.
The preferences of $q_{t,1}, q_{t,2}, p_{t,4}, p_{t,5}$ are defined as 
\begin{equation*}
q_{t,1} \colon p_{t,1} \sim p_{t,4} \sim p_{t,5}, \ \ \ 
q_{t,2} \colon p_{t,2} \sim p_{t,3} \sim p_{t,4},\ \ \
 p_{t,4} \colon q_{t,1} \sim q_{t,2}, \ \ \ 
p_{t,5} \colon q_{t,1} \sim q_{t,3}. 
\end{equation*}

Next, we consider the edges between 
the gadgets for the variables and the gadgets for the clauses. 
For each pair of integers $t \in [m]$ and $j \in \{1,2,3\}$, we define 
$q^{\ast}_{t,j}$ as follows. 
\begin{equation*}
q^{\ast}_{t,j} := 
\begin{cases}
q_{t,1} & \mbox{if $j = 1$}\\
q_{t,2} & \mbox{if $j \in \{2,3\}$}.
\end{cases} 
\end{equation*} 
If 
$\ell_{t,j} = \ell_{u,k} = \alpha_i$ and 
$(t,j) \neq (u,k)$, then 
${\sf t}_i$ is connected to 
$p_{t,j}$ and $p_{u,k}$, and
the preferences of 
${\sf t}_i, p_{t,j}, p_{u,k}$ as
\begin{equation*}
{\sf t}_i \colon p_{t,j} \sim p_{u,k} \succ d_{i,2},\ \ \ 
p_{t,j} \colon s_{t,j} \succ {\sf t}_i \succ q^{\ast}_{t,j}, \ \ \  
p_{u,k} \colon s_{u,k} \succ {\sf t}_i \succ q^{\ast}_{u,k}.
\end{equation*} 
Furthermore, if 
$\ell_{t,j} = \ell_{u,k} = \neg \alpha_i$ and 
$(t,j) \neq (u,k)$, then 
${\sf f}_i$ is connected to 
$p_{t,j}$ and $p_{u,k}$, and
the preferences of 
${\sf f}_i, p_{t,j}, p_{u,k}$ as
\begin{equation*}
{\sf f}_i \colon p_{t,j} \sim p_{u,k} \succ d_{i,2},\ \ \ 
p_{t,j} \colon s_{t,j} \succ {\sf f}_i \succ q^{\ast}_{t,j}, \ \ \ 
p_{u,k} \colon s_{u,k} \succ {\sf f}_i \succ q^{\ast}_{u,k}. 
\end{equation*} 
See Figure~\ref{fig:gadget}(c) for an example. 

In what follows, we prove that 
there exists a stable matching 
in the instance defined in the above way if and only if 
there exists a feasible solution to the given instance of 
$(3, {\rm B}2)$-{\sc sat}. 

First, we assume that there exists a stable matching $\mu$ in $G$.

\begin{lemma} \label{lemma:hardness_1}
For every integer $i \in [n]$, 
exactly one of $(d_{i,2},{\sf t}_i)$ and 
$(d_{i,2},{\sf f}_i)$ is contained in $\mu$. 
\end{lemma}
\begin{proof}
Let $i$ be an integer in $[n]$. 
Then the definition of a matching implies that 
at most one of $(d_{i,2},{\sf t}_i)$ and 
$(d_{i,2},{\sf f}_i)$ is contained in $\mu$. 
Thus, we assume that 
$(d_{i,2},{\sf t}_i) \notin \mu$ and 
$(d_{i,2}, {\sf f}_i) \notin \mu$.
If $\mu(d_{i,2}) = \emptyset$, then 
$(d_{i,2},h_{i,1})$ blocks $\mu$. 
Furthermore, if $\mu(d_{i,2}) = (d_{i,2},h_{i,1})$, then 
$(d_{i,1},h_{i,1})$ blocks $\mu$. 
Thus, in both cases, the assumption 
contradicts the fact that 
$\mu$ is stable. 
This implies that exactly one of 
$(d_{i,2},{\sf t}_i) \in \mu$ and 
$(d_{i,2}, {\sf f}_i) \in \mu$
holds. 
This completes the proof. 
\end{proof} 

We construct a solution $\phi$ to the given instance of 
$(3, {\rm B}2)$-{\sc sat} as follows.
For each integer $i \in [n]$, 
if $(d_{i,2}, {\sf t}_i) \in \mu$
(resp.\ $(d_{i,2}, {\sf f}_i) \in \mu$), 
then 
we define $\phi(\alpha_i) := 0$
(resp.\ $\phi(\alpha_i) := 1$).
Notice that 
Lemma~\ref{lemma:hardness_1}
implies that 
$\phi$ is well-defined. 
Thus, what remains is to prove that 
$\phi$ constructed in this way is a feasible solution 
to the given instance of 
$(3, {\rm B}2)$-{\sc sat}. 

\begin{lemma} \label{lemma:hardness_2} 
For every triple of integers $i \in [n]$, 
$t \in [m]$, 
and $j \in \{1,2,3\}$, 
if $(p_{t,j},{\sf t}_i) \in E$ and 
$(p_{t,j},q^{\ast}_{t,j}) \in \mu$
{\rm (}resp.\ 
$(p_{t,j},{\sf f}_i) \in E$ and 
$(p_{t,j},q^{\ast}_{t,j}) \in \mu${\rm )}, then 
$(d_{i,2},{\sf f}_i) \in \mu$
{\rm (}resp.\ $(d_{i,2},{\sf t}_i) \in \mu${\rm )}. 
\end{lemma}
\begin{proof} 
For every triple of integers $i \in [n]$, 
$t \in [m]$, 
and $j \in \{1,2,3\}$, 
if $(p_{t,j},{\sf t}_i) \in E$,
$(p_{t,j},q^{\ast}_{t,j}) \in \mu$,
and $(d_{i,2},{\sf t}_i) \in \mu$, then 
$(p_{t,j},{\sf t}_i)$ blocks $\mu$. 
However, this contradicts the fact that 
$\mu$ is stable. 
The remaining statement can be proved in the same way.
This completes the proof. 
\end{proof} 

\begin{lemma} \label{lemma:hardness_3}
For every integer $t \in [m]$, 
exactly one of 
$(p_{t,1},q_{t,1}), (p_{t,2},q_{t,2}), (p_{t,3},q_{t,2})$
is contained in $\mu$. 
\end{lemma}
\begin{proof}
The definition of the preference of $p_{t,5}$ implies that 
$(p_{t,5},q_{t,3}) \in \mu$. Thus, 
the definition of the preference of $q_{t,1}$ implies that
exactly one of $(p_{t,1},q_{t,1})$ and $(p_{t,4},q_{t,1})$
is contained in $\mu$. 
If $(p_{t,1},q_{t,1}) \in \mu$, then 
the definition of the preference of $q_{t,1}$ implies that 
$(p_{t,4},q_{t,2}) \in \mu$, and 
the proof is done. 
If $(p_{t,4},q_{t,1}) \in \mu$, then 
the definition of the preference of $p_{t,4}$ implies that
exactly one of $(p_{t,2},q_{t,2})$ and $(p_{t,3},q_{t,2})$
is contained in $\mu$. 
This completes the proof. 
\end{proof} 

It follows from Lemma~\ref{lemma:hardness_2} that, 
for every integer $t \in [m]$, 
if $(p_{t,j},q^{\ast}_{t,j}) \in \mu$, then  
$1$ is assigned to $\ell_{t,j}$ under $\phi$.  
This and 
Lemma~\ref{lemma:hardness_3} imply that 
$\phi$ is a feasible solution 
to the given instance of 
$(3, {\rm B}2)$-{\sc sat}.

Next we assume that 
there exists a feasible solution $\phi$ 
to the given instance of $(3, {\rm B}2)$-{\sc sat}. 
Then we construct the matching $\mu$ as follows. 

For each integer $i \in [n]$, if 
$\phi(\alpha_i) = 0$, then 
$\mu$ contains 
$(d_{i,1},h_{i,1})$ and $(d_{i,2},{\sf t}_i)$.
On the other hand, for each integer $i \in [n]$, if 
$\phi(\alpha_i) = 1$, then 
$\mu$ contains 
$(d_{i,1},h_{i,1})$ and $(d_{i,2},{\sf f}_i)$.

Since $\phi$ is a feasible solution to 
the given instance of $(3, {\rm B}2)$-{\sc sat}, 
for every integer $t \in [m]$, 
there exists an integer $j_t \in \{1,2,3\}$ such that 
$1$ is assigned to $\ell_{t,j_t}$ under $\phi$.
(If there exist more than one such an integer $j$, then 
we arbitrarily select one such an integer $j$.) 
For each integer $t \in [m]$,
if $j_t = 1$, then $\mu$ contains the edges described by bold lines in 
Figure~\ref{fig:construction_matching}(a). 
For each integer $t \in [m]$,
if $j_t = 2$, then $\mu$ contains the edges described by bold lines in 
Figure~\ref{fig:construction_matching}(b). 
Finally, for each integer $t \in [m]$,
if $j_t = 3$, then $\mu$ contains the edges described by bold lines in 
Figure~\ref{fig:construction_matching}(c). 
See Figure~\ref{fig:example} for an example.

\begin{figure*}[t]
\begin{minipage}{0.33\hsize}
\centering
\includegraphics[width=4cm]{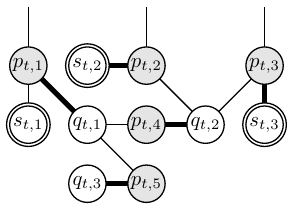}
\par(a)
\end{minipage}
\begin{minipage}{0.32\hsize}
\centering
\includegraphics[width=4cm]{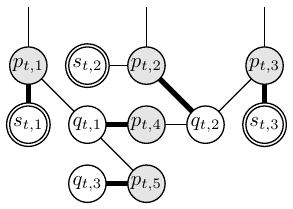}
\par(b)
\end{minipage}
\begin{minipage}{0.33\hsize}
\centering
\includegraphics[width=4cm]{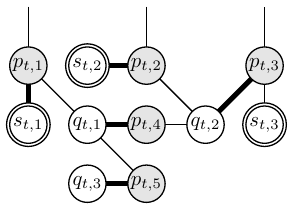}
\par(c)
\end{minipage}
\caption{(a) $j_t = 1$. 
(b) $j_t = 2$.
(c) $j_t = 3$.} 
\label{fig:construction_matching}
\end{figure*} 

\begin{figure*}[t]
\begin{minipage}{0.49\hsize}
\centering
\includegraphics[width=6cm]{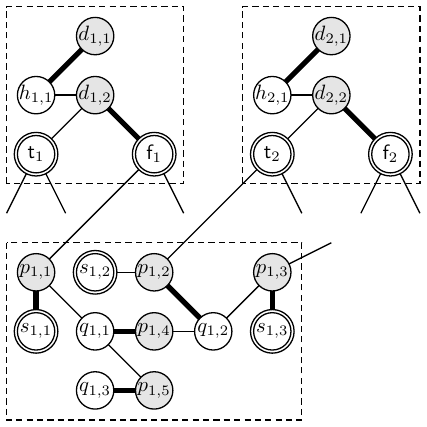}
\par(a)
\end{minipage}
\begin{minipage}{0.49\hsize}
\centering
\includegraphics[width=6cm]{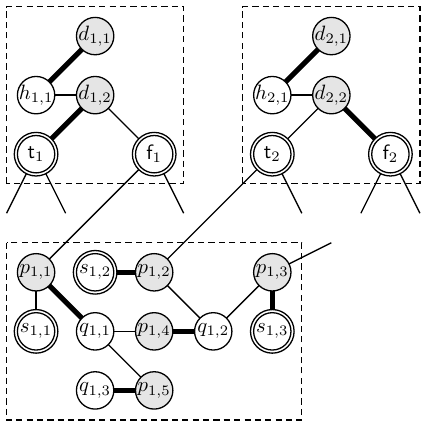}
\par(b)
\end{minipage}
\caption{(a) $\phi(\alpha_1) = 1$, $\phi(\alpha_2) = 1$, and $j_1 = 2$. 
(b) $\phi(\alpha_1) = 0$, $\phi(\alpha_2) = 1$, and $j_1 = 1$.}
\label{fig:example}
\end{figure*}

What remains is to prove that $\mu$ is stable. 
It is not difficult to see that 
there does not exist an edge 
in the gadgets that blocks $\mu$. 
Thus, we consider the edges between 
the gadgets. 
Let us consider the vertex $p_{t,j}$.
Assume that 
$p_{t,j}$ is connected to ${\sf t}_i$. 
Then we prove that 
$(p_{t,j},{\sf t}_i)$ does not block $\mu$.
If $(p_{t,j}, s_{t,j}) \in \mu$, then 
the preference of $p_{t,j}$ implies that 
$(p_{t,j},{\sf t}_i)$ does not block $\mu$.
Thus, we assume that $(p_{t,j},q^{\ast}_{t,j}) \in \mu$. 
In this case, 
the definition of $\mu$ implies that 
$(d_{i,2},{\sf f}_i) \in \mu$.
Thus, since ${\sf t}_i \in S$, 
$(p_{t,j},{\sf t}_i)$ does not block $\mu$.
In the same way, we can prove the case where 
$p_{t,j}$ is connected to ${\sf f}_i$.
This completes the proof. 
\end{proof}

\section{Useful Lemmas} 

In this section, we give lemmas that are needed to propose 
algorithms for special cases 
where the strongly stable matching problem with closures
can be solved in polynomial time. 

For each vertex $v \in D \cup H$ and each subset $F \subseteq E$, 
we define ${\rm Ch}_v(F)$ as the set of edges $e \in F(v)$ such that 
$e \succsim_v f$ for every edge $f \in F(v)$. 
For each subset $F \subseteq E$, 
we define 
\begin{equation*}
{\rm Ch}_D(F) := 
\bigcup_{d \in D}{\rm Ch}_d(F), \ \ \ 
{\rm Ch}_H(F) := 
\bigcup_{h \in H}{\rm Ch}_h(F), \ \ \ 
{\rm Ch}(F) := {\rm Ch}_H({\rm Ch}_D(F)). 
\end{equation*}

For each subset $F \subseteq E$, $F$ is said to be \emph{flat} 
if $e \sim_v f$
for every vertex $v \in D \cup H$ and 
every pair of edges $e,f \in F(v)$.
For example, for every subset $F \subseteq E$, 
${\rm Ch}(F)$ is flat. 

Let $F$ be a flat subset of $E$.
Let $e = (d,h)$ be 
an edge in $E \setminus F$ such that 
$F(h) \neq \emptyset$. 
Then 
we write $e \succ_d F$ if 
one of the following conditions are satisfied. 
\begin{itemize}
\item
$F(d) = \emptyset$.
\item
$F(d) \neq \emptyset$ and 
$e \succ_d f$ for an edge $f \in F(d)$.
\end{itemize}
Furthermore, we write $e \sim_d F$
if $F(d) \neq \emptyset$ and 
$e \sim_d f$  
for an edge $f \in F(d)$.
If at least one of 
$e \succ_d F$ and 
$e \sim_d F$, then 
we write $e \succsim_d F$. 
We write 
$e \succsim_h F$ (resp.\ $e \succ_h F$)
if 
$e \succsim_h f$ (resp.\ $e \succ_h f$)
for an edge $f \in F(h)$.

For each flat subset $F \subseteq E$, we define 
${\sf block}(F)$ as the set of edges $e = (d,h) \in E \setminus F$
such that the following conditions are satisfied.
\begin{itemize}
\item
$F(h) \neq \emptyset$.
\item
$e \succsim_d F$.
\item
If $e \succ_d F$, then $e \succsim_h F$. 
\item
If $e \sim_d F$, then $e \succ_h F$. 
\end{itemize}

In what follows, the following lemma plays an important role. 
Although the setting of this paper and the setting in \cite{Irving94,Manlove99} 
are different, 
this lemma can be proved by using an idea similar to the 
idea used in \cite{Irving94,Manlove99}. 
In Section~\ref{section:pre-process}, 
we prove Lemma~\ref{lemma:pre-process} by 
giving a polynomial-time algorithm for finding 
a subset $R \subseteq E$ and a matching $\mu$ in $G$
satisfying the conditions in Lemma~\ref{lemma:pre-process}. 

\begin{lemma} \label{lemma:pre-process}
For every instance of the strongly stable matching problem with 
closures, there exist a subset $R \subseteq E$ and a matching $\mu$ in $G$
satisfying the following 
conditions.
\begin{description}
\item[\bf (R1)]
For any edge $e \in R$, there does not exist a stable matching 
in $G$ containing $e$. 
\item[\bf (R2)] 
$\mu \subseteq {\rm Ch}(E \setminus R)$.
\item[\bf (R3)] 
$\mu(d) \neq \emptyset$ holds for every doctor $d \in D[E \setminus R]$.
\item[\bf (R4)] 
$R \cap {\sf block}({\rm Ch}(E \setminus R)) = \emptyset$. 
\item[\bf (R5)] 
For every doctor $d \in D$ and every pair of 
edges $e \in R(d)$ and $f \in E(d) \setminus R$, 
$e \succsim_d f$. 
\end{description}
\end{lemma}

\begin{lemma} \label{lemma:property_R}
Let $R$ and $\mu$ be 
a subset of $E$ and a matching in $G$
satisfying the conditions in Lemma~\ref{lemma:pre-process}, respectively. 
Define $L := {\rm Ch}(E \setminus R)$.
Then the following statements hold. 
\begin{description}
\item[(i)]
For every doctor $d \in D$ such that $\mu(d) = \emptyset$, 
we have $L(d) = \emptyset$. 
\item[(ii)]
For every doctor $d \in D$ such that $L(d) = \emptyset$, 
we have $E(d) \subseteq R$. 
\item[(iii)] 
For every edge $e = (d,h) \in E \setminus R$, 
we have $\mu(d) \neq \emptyset$. 
\end{description}
\end{lemma}
\begin{proof}
First, we prove the statement (i).
Let $d$ be a doctor in $D$ such that 
$\mu(d) = \emptyset$. 
Then (R3) implies that 
$E(d) \setminus R = \emptyset$. 
Thus,  
since $L(d) \subseteq {\rm Ch}_d(E \setminus R) = \emptyset$, 
we have $L(d) = \emptyset$. 

Next, we prove the statement (ii).
Let $d$ be a doctor in $D$ such that 
$L(d) = \emptyset$. 
Assume that 
$E(d) \not\subseteq R$. 
Then (R3) implies that $\mu(d) \neq \emptyset$. 
Furthermore, (R2) implies that 
$\mu(d) \in L(d)$. 
This contradicts the fact that $L(d) = \emptyset$. 
This completes the proof. 

Finally, we prove the statement (iii).
Since $d \in D[E \setminus R]$, (R3) implies this 
statement. 
\end{proof} 

For each subset $R \subseteq E$, each matching $\mu$ in $G$ 
such that $\mu \subseteq {\rm Ch}(E \setminus R)$, and each hospital $h \in H \setminus S$, 
$h$ is called a \emph{critical hospital with respect to $(R, \mu)$} 
if $\mu(h) = \emptyset$ and 
at least one of 
${\rm Ch}(E \setminus R) \cap E(h) \neq \emptyset$
and $R(h) \neq \emptyset$ holds. 

\begin{lemma} \label{lemma:no_critical}
Let $R$ and $\mu$ be 
a subset of $E$ and a matching in $G$
satisfying the conditions in Lemma~\ref{lemma:pre-process}, respectively. 
If there does not exist a critical hospital with respect to $(R,\mu)$, then 
$\mu$ is a stable matching in $G$ 
such that, for every stable matching $\sigma$ in $G$
and every doctor $d \in D$, 
we have $\mu(d) \succsim_d \sigma(d)$.  
\end{lemma}
\begin{proof}
Assume that 
there does not exist a critical hospital with respect to $(R,\mu)$.

\begin{claim} \label{claim:claim_1:lemma:no_critical}
Assume that $\mu$ is stable.
Let $\sigma$ and $d$ be a stable matching in $G$ and 
a doctor in $D$, respectively. 
Then $\mu(d) \succsim_d \sigma(d)$. 
\end{claim}  
\begin{proof}
It follows from (R1) that 
$\sigma \subseteq E\setminus R$. 
If $\sigma(d) = \emptyset$, then 
$\mu(d) \succsim_d \sigma(d)$.
Thus, we consider the case where $\sigma(d) \neq \emptyset$. 
If $E(d) \subseteq R$, then 
$\sigma(d) = \emptyset$. Thus, 
$E(d) \not\subseteq R$.
In this case, (R2) and (R3) imply that 
$\mu(d) \in {\rm Ch}_d(E \setminus R)$. 
Thus, since $\sigma(d) \in E\setminus R$,  
$\mu(d) \succsim_d \sigma(d)$. 
\end{proof}

Claim~\ref{claim:claim_1:lemma:no_critical} 
implies that what remains is to prove that $\mu$ is stable. 
Define $L := {\rm Ch}(E \setminus R)$.

In what follows, we fix an edge $e = (d,h) \in E \setminus \mu$, and 
we prove that $e$ does not block $\mu$.
If $\mu(h) = \emptyset$ and $h \in S$, then 
$e$ does not block $\mu$.
Thus, when $\mu(h) = \emptyset$, 
it is sufficient to consider the case where 
$h \in H \setminus S$.  

First, we consider the case where $e \succ_d \mu(d)$. 
If $\mu(h) \succ_h e$, then $e$ does not block $\mu$.
Thus, 
we assume that $e \succsim_h \mu(h)$, and 
we derive a contradiction. 
We first prove that $e \in R$. 
If $e \notin R$, 
then since $e \in E(d) \setminus R$, 
(R3) implies that  
$\mu(d) \neq \emptyset$. 
Thus, (R2) implies that 
$\mu(d) \succsim_d e$.
However, this contradicts the assumption that 
$e \succ_d \mu(d)$.
Thus, $e \in R$.
Notice that since $e \in R$, 
we have $e \notin L$. 
We continue the proof of 
this case. 
\begin{itemize}
\item
Assume that $\mu(h) \neq \emptyset$. 
Then since $\mu \subseteq L$, $L(h) \neq \emptyset$. 
If $\mu(d) \neq \emptyset$, then 
since $\mu(d) \in L(d)$ and $e \succ_d \mu(d)$, 
we have $e \succ_d L$. 
If $\mu(d) = \emptyset$, then 
since Lemma~\ref{lemma:property_R} implies that 
$L(d) = \emptyset$, we have $e \succ_d L$. 
Thus, in both cases, $e \succ_d L$.
Furthermore, since $\mu(h) \in L$ and 
$e \succsim_h \mu(h)$, 
we have $e \succsim_h L$. 
This 
implies that 
$e \in R \cap {\sf block}(L)$.
However, this contradicts (R4).  
\item
Assume that 
$\mu(h) = \emptyset$. 
In this case, since $h \in H \setminus S$ and $e \in R(h)$, 
$e$ is a critical hospital with respect to $(R,\mu)$. 
This contradicts the assumption of this lemma.
\end{itemize}
Thus, in the case where 
$e \succ_d \mu(d)$, this lemma holds. 

Next, we consider the case where $e \sim_d \mu(d)$. 
If $\mu(h) \succsim_h e$, then $e$ does not block $\mu$.
Thus, we assume that 
$e \succ_h \mu(h)$, and 
we derive a contradiction.  
We first prove that
$e \in R$. 
If $e \notin R$, then 
since $e \sim_d \mu(d)$, 
$e \in {\rm Ch}_D(E \setminus R)$. 
Thus, since 
$e \in {\rm Ch}_D(E \setminus R) \cap E(h)$, 
${\rm Ch}(E \setminus R) \cap E(h) \neq \emptyset$. 
If $\mu(h) = \emptyset$, then 
$h$ is a critical hospital with respect to $(R,\mu)$. 
This contradicts the assumption of this lemma.
If $\mu(h) \neq \emptyset$, then 
since (R2) implies that 
$\mu(h) \in {\rm Ch}(E \setminus R)$, we have 
$\mu(h) \succsim_h e$. 
However, this contradicts the assumption that 
$e \succ_h \mu(h)$. 
Thus, $e \in R$.
Notice that since $e \in R$, 
we have $e \notin L$. 
We continue the proof of 
this case.
\begin{itemize}
\item
Assume that 
$\mu(h) \neq \emptyset$. 
Since $\mu \subseteq L$, we have $L(h) \neq \emptyset$. 
If $\mu(d) \neq \emptyset$, then 
since $\mu(d) \in L(d)$ and $e \sim_d \mu(d)$, 
we have $e \sim_d L$. 
If $\mu(d) = \emptyset$, then 
since $L(d) = \emptyset$, $e \succ_d L$.
Thus, in both cases, we have 
$e \succsim_d L$.
Furthermore, since $\mu(h) \in L$ and 
$e \succ_h \mu(h)$, 
we have $e \succ_h L$. 
This implies that 
$e \in R \cap {\sf block}(L)$.
However, this contradicts (R4).  
\item
Assume that 
$\mu(h) = \emptyset$. 
Since $h \in H \setminus S$
and $e \in R(h)$, 
$e$ is a critical hospital with respect to $(R,\mu)$. 
This contradicts 
the assumption of this lemma.
\end{itemize}
Thus, in the case where 
$e \sim_d \mu(d)$, this lemma holds. 
This completes the proof.
\end{proof}

\section{Separated Preferences} 
\label{section:separated} 

In this section, 
we consider an instance of 
the strongly stable matching problem with closures 
satisfying the following condition. 
\begin{description}
\item[($\star$)]
For every doctor $d \in D$ and 
every pair of edges $(d,h),(d,h^{\prime}) \in E(d)$ such that 
$h \in S$ and $h^{\prime} \notin S$, 
we have $(d,h^{\prime}) \succ_d (d,h)$.  
\end{description}
It should be noted that if we interchange the roles of $S$ and $H \setminus S$, then 
the problem becomes NP-complete (see Theorem~\ref{theorem:hardness}). 
This special case can be regarded as  
a common generalization of the strongly stable matching problem 
and the envy-free matching problem (see Section~\ref{section:envy}). 

Our algorithm is described in Algorithm~\ref{alg:separate}. 
Since we can find a subset $R \subseteq E$ and a matching $\mu$ in $G$
satisfying the conditions in Lemma~\ref{lemma:pre-process} in polynomial time 
(see Section~\ref{section:pre-process}), 
Algorithm~\ref{alg:separate} is a polynomial-time 
algorithm. 
\begin{algorithm}[h]
Find a subset $R \subseteq E$ and a matching $\mu$ in $G$
satisfying the conditions in Lemma~\ref{lemma:pre-process}.\\
\If{there exists a critical hospital with respect to $(R,\mu)$}
  {
     Output {\bf No}, and halt. 
  }
Output $\mu$, and halt. 
\caption{Algorithm under the assumption ($\star$)}
\label{alg:separate}
\end{algorithm}

\begin{lemma} \label{lemma:separate_No} 
If Algorithm~\ref{alg:separate} outputs {\bf No}, 
then there does not exist a stable matching in $G$.
\end{lemma}
\begin{proof}
Since Algorithm~\ref{alg:separate} outputs {\bf No}, 
there exists a critical hospital $h^{\ast}$ with respect to $(R,\mu)$. 
The assumption ($\star$) implies that, for every doctor $d \in D[E \setminus R]$, 
exactly one of the following conditions is satisfied. 
\begin{description}
\item[(i)]
$h \in H \setminus S$ for every edge $(d,h) \in {\rm Ch}_d(E \setminus R)$. 
\item[(ii)]
$h \in S$ for every edge $(d,h) \in {\rm Ch}_d(E \setminus R)$. 
\end{description} 
Define $D^+$ (resp.\ $D^-$) as the set of doctors $d \in D[E\setminus R]$ 
satisfying the condition (i) (resp.\ (ii)).  
Then the assumption ($\star$) implies that, 
for every doctor $d \in D^-$ and every edge $(d,h) \in E \setminus R$, 
we have $h \in S$. 

Define $H^+ := \Gamma_{\mu}(D^+)$. 
Then (R2) and the definition of $D^+$ 
imply that 
$H^+ \subseteq H \setminus S$, and (R3) implies that 
$|H^+| = |D^+|$. 
Since
$\mu(h) \neq \emptyset$ holds for every hospital $h \in H^+$ 
and $\mu(h^{\ast}) = \emptyset$, 
we have $h^{\ast} \notin H^+$. 
Define $H^{\ast} := H^+ \cup \{h^{\ast}\}$. 
Notice that $|H^{\ast}| > |H^+|$. 

\begin{claim} \label{claim:separate_No_1}
Assume that there exists a 
stable matching $\sigma$ in $G$. 
Let $h$ be a hospital in $H \setminus S$ 
such that $\sigma(h) \neq \emptyset$. 
Furthermore, we assume that 
$\sigma(h) = (d,h)$.
Then we have $d \in D^+$. 
\end{claim}
\begin{proof}
Since (R1) implies that 
$\sigma \subseteq E \setminus R$, 
we have $(d,h) \in E \setminus R$, and 
thus $d \in D[E \setminus R]$.  
This implies that what remains is to prove that 
$d \notin D^-$. 
Recall that,  
for every doctor $d^{\prime} \in D^-$ and every 
edge $(d^{\prime},h^{\prime}) \in E \setminus R$, 
$h^{\prime} \in S$.
Thus, since $(d,h) \in E \setminus R$, 
if $d \in D^-$, then 
$h \in S$. 
However, this contradicts the fact that $h \in H \setminus S$.
Thus, $d \notin D^-$. 
This completes the proof. 
\end{proof} 

In what follows, we assume that there exists a stable matching $\sigma$ in $G$. 
Then we derive a contradiction. 
Notice that (R1) implies that $\sigma \subseteq E \setminus R$. 

First, we consider the case where
${\rm Ch}(E \setminus R) \cap E(h^{\ast}) \neq \emptyset$. 
Since $h^{\ast} \in H \setminus S$, 
$H^{\ast} \subseteq H \setminus S$.
Thus, 
since 
$|H^{\ast}| > |H^+| = |D^+|$, Claim~\ref{claim:separate_No_1}
implies that 
there exists a hospital $h \in H^{\ast}$ such that 
$\sigma(h) = \emptyset$. 
Define 
the edge $e$ as follows. 
If $h \in H^+$, then 
we define $e := \mu(h)$. 
If $h = h^{\ast}$, then 
we define $e$ as an 
edge 
in ${\rm Ch}(E \setminus R) \cap E(h^{\ast})$.
Assume that 
$e = (d, h)$. 
Then since $\sigma \subseteq E\setminus R$ and 
$e \in {\rm Ch}_d(E \setminus R)$, 
we have $e \succsim_d \sigma(d)$. 
Thus, since $\sigma(h) = \emptyset$ and $h \in H \setminus S$, 
$e$ blocks $\sigma$.
However, this contradicts the fact that $\sigma$ is stable. 

Next, we consider the case where
$R(h^{\ast}) \neq \emptyset$. 
Let $e^{\ast} = (d^{\ast},h^{\ast})$ be an edge in $R(h^{\ast})$. 
Then we first prove that 
$\sigma(h^{\ast}) \neq \emptyset$. 
Assume that 
$\sigma(h^{\ast}) = \emptyset$. 
Then since $\sigma \subseteq E \setminus R$, 
(R5) implies that 
$e^{\ast} \succsim_{d^{\ast}} \sigma(d^{\ast})$. 
Thus, $e^{\ast}$ blocks $\sigma$. 
However, this contradicts the fact that $\sigma$ is stable. 
Thus, this completes the proof of 
$\sigma(h^{\ast}) \neq \emptyset$. 
Assume that 
$\sigma(h^{\ast}) = (d^{\prime},h^{\ast})$. 
Since $h^{\ast} \in H \setminus S$, Claim~\ref{claim:separate_No_1}
implies that 
$d^{\prime} \in D^+$. 
Thus, since 
$H^{\ast} \subseteq H \setminus S$,  
$|H^{\ast}| > |H^+| = |D^+|$, and 
$\sigma(h^{\ast}) \neq \emptyset$, 
it follows from Claim~\ref{claim:separate_No_1}
that 
there exists a hospital $h \in H^+$ such that 
$\sigma(h) = \emptyset$. 
Define $e := \mu(h)$. 
Assume that 
$e = (d, h)$. 
Since $\sigma \subseteq E\setminus R$ and 
$e \in {\rm Ch}_d(E \setminus R)$, 
we have $e \succsim_d \sigma(d)$. 
Thus, since $\sigma(h) = \emptyset$
and $h \in H \setminus S$, 
$e$ blocks $\sigma$.
However, 
this contradicts the fact that $\sigma$ is stable. 
\end{proof} 

\begin{theorem} \label{theorem:separate} 
Algorithm~\ref{alg:separate} can correctly solve 
the strongly stable matching problem with closures under the assumption 
{\rm ($\star$)}.
\end{theorem}
\begin{proof}
If Algorithm~\ref{alg:separate} output a matching $\mu$ in $G$, 
then Lemma~\ref{lemma:no_critical} implies that 
$\mu$ is stable.
Thus, this theorem follows from Lemma~\ref{lemma:separate_No}.
\end{proof} 

\subsection{Strongly stable matchings and envy-free matchings}
\label{section:envy} 

In this section, we consider relations between the special case in Theorem~\ref{theorem:separate} 
and existing models. 

First, we consider the strongly stable matching problem
considered in \cite{Irving94,Manlove99}.
If $S = \emptyset$, then 
our problem is the same as the strongly stable matching problem.
Notice that if $S = \emptyset$, then 
the assumption ($\star$) clearly holds. 

Next, we consider the 
envy-free matching 
problem considered in \cite{GanSV19}. 
The envy-free matching problem is defined as follows. 
In this problem,  
we are given a finite simple undirected bipartite graph 
$G = (D, H; E)$. 
For each doctor 
$d \in D$, we are given 
a transitive binary relation 
$\gtrsim_d$ on $\Gamma_E(d)$
such that, for every pair of hospitals $h,h^{\prime} \in \Gamma_E(d)$, 
at least one of $h \gtrsim_d h^{\prime}$ and $h^{\prime} \gtrsim_d h$ holds.
For each doctor $d \in D$ and each pair of 
hospitals $h,h^{\prime} \in \Gamma_E(d)$, 
if $h \gtrsim_d h^{\prime}$ and $h^{\prime} \not\gtrsim_d h$, then
we write $h >_d h^{\prime}$. 
A matching $\mu$ in $G$ is said to be \emph{$D$-perfect}
if $|\mu| = |D|$. 
For each $D$-perfect matching $\mu$ in $G$ and 
each doctor $d \in D$, we define ${\sf p}_{\mu}(d)$ 
as the unique hospital $h \in H$ such that 
$(d,h) \in \mu$. 
A $D$-perfect matching $\mu$ in $G$ 
is said to be \emph{envy-free} if 
there does not exist a pair of distinct doctors 
$d,d^{\prime} \in D$ such that 
${\sf p}_{\mu}(d^{\prime}) >_d {\sf p}_{\mu}(d)$. 
The goal of the envy-free 
matching problem is to determine 
whether there exists a $D$-perfect envy-free matching in $G$ and to find one 
if exists.  
Gan, Suksompong, and Voudouris~\cite{GanSV19} proved that 
the envy-free matching problem can be solved in polynomial time. 

The envy-free matching problem can be reduced to our problem as follows. 
For each doctor $d \in D$,
we define $\succsim_d$ in such a way that, 
for each pair of edges 
$(d,h),(d,h^{\prime}) \in E(d)$, 
$(d,h) \succsim_d (d,h^{\prime})$ if and only if 
$h \gtrsim_d h^{\prime}$.
For each hospital $h \in H$, we define 
$\succsim_h$ in such a way that 
$e \sim_h f$ for 
every pair of edges $e, f \in E(h)$. 
Define $S := H$. 
Notice that if $S = H$, then 
the assumption ($\star$) clearly holds. 
Furthermore, in this setting, Algorithm~\ref{alg:separate} always outputs 
a matching $\mu$. 

It is not difficult to see that, 
for every $D$-perfect matching $\sigma$ in $G$ 
and every pair of doctors $d,d^{\prime} \in D$, 
${\sf p}_{\sigma}(d^{\prime}) >_d {\sf p}_{\sigma}(d)$
if and only if 
$(d,{\sf p}_{\sigma}(d^{\prime}))$ blocks $\sigma$
in the reduced instance. 
Thus, 
if there 
exists a $D$-perfect envy-free matching $\sigma$ in $G$, 
then 
$\sigma$ is also a $D$-perfect stable matching in $G$.
If the output $\mu$ of Algorithm~\ref{alg:separate} is not $D$-perfect, then 
Lemma~\ref{lemma:no_critical} implies that 
there exists a doctor $d \in D$ such that, for every stable matching 
$\sigma$ in $G$, we have $\sigma(d) = \emptyset$. 
This implies that 
if the output $\mu$ of Algorithm~\ref{alg:separate} is not $D$-perfect,
then 
there does not exist a $D$-perfect stable matching in $G$, i.e., 
there does not exist a $D$-perfect envy-free matching in $G$. 
In addition, 
if $\mu$ is $D$-perfect, then 
since $\mu$ is stable, 
$\mu$ is also a $D$-perfect envy-free matching in $G$.

\section{Bounded-Degree Instances} 
\label{section:bounded_degree} 

In this section, 
we consider an instance of 
the strongly stable matching problem with closures 
such that 
$|E(d)| \le 2$ holds for every doctor $d \in D$. 
The goal of this section is to prove that
this special case can be solved in polynomial time. 

First, we find a subset $R \subseteq E$ and a matching $\mu$ in $G$
satisfying the conditions in Lemma~\ref{lemma:pre-process}.
If there does not exist a critical hospital 
with respect to $(R,\mu)$, then Lemma~\ref{lemma:no_critical} 
implies that $\mu$ is a stable matching in $G$. 
Thus, in what follows, we assume that 
there exists a critical hospital 
with respect to $(R,\mu)$. 
Define $L := {\rm Ch}(E \setminus R)$. 
Notice that (R2) implies that $\mu \subseteq L$. 

Define the undirected bipartite graph 
$G^{\ast}$ by $G^{\ast} := (D,H;L)$. 
For each subset $X \subseteq D \cup H$, 
$X$ is called a \emph{connected component} in $G^{\ast}$
if $X$
is an inclusion-wise maximal subset $X \subseteq D \cup H$ such that 
$G^{\ast}\langle X \rangle$ is connected.
(Recall that 
$G^{\ast}\langle X \rangle$ is 
the subgraph of $G^{\ast}$ induced by $X$.) 
For each connected component $X$ in $G^{\ast}$, 
we define $D_X := D \cap X$ and 
$H_X := H \cap X$. 
We define $\mathcal{C}$ as the set of 
connected components $X$ in $G^{\ast}$ such that 
$|X| \ge 2$. 

\begin{lemma} \label{lemma:size_component}
For 
every element $X \in \mathcal{C}$, 
we have $|D_X| \le |H_X|$. 
\end{lemma}
\begin{proof}
Let $X$ be an element in $\mathcal{C}$. 
Then Lemma~\ref{lemma:property_R} implies that, for every doctor $d \in D_X$, 
since $L(d) \neq \emptyset$, 
we have $\mu(d) \neq \emptyset$.
Since $\mu \subseteq L$,  
$\Gamma_{\mu}(D_X) \subseteq H_X$. 
Thus, since $\mu$ is a matching in $G$, 
$|D_X| \le |H_X|$. 
This completes the proof. 
\end{proof} 

Define $\mathcal{P}$ (resp.\ $\mathcal{Q}$) as the set of 
elements $X \in \mathcal{C}$ such that 
$|D_X| = |H_X|$
(resp.\ $|D_X| < |H_X|$).
Then Lemma~\ref{lemma:size_component} implies that 
$\mathcal{C} = \mathcal{P} \cup \mathcal{Q}$. 
Define $\mathcal{R}$ as the set of 
connected components $X$ in $G^{\ast}$ such that 
$X$ consists of a single hospital $h \in H \setminus S$ and 
$R(h) \neq \emptyset$.

For each element $X \in \mathcal{P} \cup \mathcal{Q}$, 
we define the undirected graph $G_X$ as follows. 
The vertex set of $G_X$ is $H_X$. 
For each doctor $d \in D_X$ such that 
$|L(d)| = 2$, 
$G_X$ contains an edge $e_d$ connecting 
$h$ and $h^{\prime}$, where we assume that 
$L(d) = \{(d,h),(d,h^{\prime})\}$.
(Notice that $G_X$ may contain parallel edges.) 
Then 
it is not difficult to see that, 
for every element $X \in \mathcal{P} \cup \mathcal{Q}$, 
since $X$ is a connected component in $G^{\ast}$, 
$G_X$ is connected. 
Thus, for every element $X \in \mathcal{P} \cup \mathcal{Q}$, 
since $G_X$ contains a spanning tree as a subgraph, 
\begin{equation} \label{eq:family}
|H_X| - 1 \le |D_X| - |{\sf Leaf}_L(D_X)|. 
\end{equation}
where we define ${\sf Leaf}_L(D_X) := \{d \in D_X \mid |L(d)| = 1\}$. 
Notice that the right-hand side of 
\eqref{eq:family} is the number of edges of $G_X$. 

\begin{lemma} \label{lemma:family_p}
For every element $X \in \mathcal{P}$, 
$|{\sf Leaf}_L(D_X)| \le 1$.
\end{lemma}
\begin{proof}
For every element $X \in \mathcal{P}$, 
since $|D_X| = |H_X|$, this lemma follows from \eqref{eq:family}. 
\end{proof} 

We define $\mathcal{P}^{\ast}$ as the family of elements $X \in \mathcal{P}$
such that $|{\sf Leaf}_L(D_X)| = 1$. 
For each element $X \in \mathcal{P}^{\ast}$, 
we define $d_X$ as the doctor in ${\sf Leaf}_L(D_X)$.
For each element $X \in \mathcal{P}^{\ast}$, 
we define $h_X$ as the hospital in $H_X$ 
such that $(d_X,h_X) \in L$.
Notice that,
for every element $X \in \mathcal{P}^{\ast}$, we have 
$(d_X,h_X) \in \mu$. 
Furthermore, 
for each element $X \in \mathcal{P}^{\ast}$, if
$|E(d_X)| = 2$, then 
we define $\overline{h}_X$ as the hospital 
in $H$ such that 
$\overline{h}_X \neq h_X$ and 
$(d_X, \overline{h}_X) \in E$. 

\begin{lemma} \label{lemma:family_q}
For every element $X \in \mathcal{Q}$, 
the following statements hold.
\begin{description}
\item[(i)]
$|L(d)| = 2$ for every doctor $d \in D_X$.
\item[(ii)] 
$|H_X| = |D_X|+1$. 
\end{description}
\end{lemma}
\begin{proof}
Let $X$ be an element in $\mathcal{Q}$. 
Since $|D_X| < |H_X|$, 
\eqref{eq:family} implies that 
${\sf Leaf}_L(D_X) = \emptyset$. 
This implies 
the statement (i).
Since ${\sf Leaf}_L(D_X) = \emptyset$, 
\eqref{eq:family} implies that 
$|H_X| \le |D_X| + 1$. 
Thus, since 
$|D_X| < |H_X|$ (i.e., $|D_X|+1 \le |H_X|$), 
we have $|D_X|+1 = |H_X|$.
This completes the proof. 
\end{proof} 

\begin{lemma} \label{lemma:tree} 
For every element $X \in \mathcal{Q}$ and 
every hospital $h \in H_X$, 
there exists a matching $\xi$ in $G$ such that 
$\xi \subseteq L(D_X)$, 
$\xi(h) = \emptyset$, and 
$\xi(d) \neq \emptyset$ holds for every doctor $d \in D_X$.
\end{lemma}
\begin{proof}
Let $X$ and $h$ be an element in $\mathcal{Q}$ and a hospital 
in $H_X$, respectively. 
Then since 
$G_X$ is connected and  
Lemma~\ref{lemma:family_q} implies that 
$|H_X| = |D_X|+1$, 
$G_X$ is a tree. 
Here we regard $G_X$ as the rooted tree rooted at $h$.
Then for each doctor $d \in D_X$, we define 
$h_d$ as follows.
Let $h,h^{\prime}$ be the hospitals in 
$H_X$ such that 
$L(d) = \{(d,h),(d,h^{\prime})\}$. 
Assume that 
$h^{\prime}$ is a child of $h$ in the rooted tree $G_X$. 
Then 
we define $h_d := h^{\prime}$. 
Define the matching $\xi$ by 
defining $\xi(d) := (d,h_d)$ for each doctor $d \in D_X$
(see Figure~\ref{fig:tree}). 
Then it is not difficult to see that 
$\xi$ is a matching satisfying the conditions in this lemma.
This completes the proof. 
\end{proof} 

\begin{figure}[t]
\begin{minipage}{0.33\hsize}
\centering
\includegraphics[width=5.5cm]{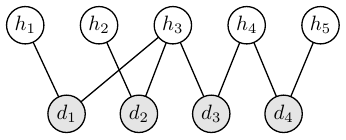}
\par(a)
\end{minipage}
\begin{minipage}{0.32\hsize}
\centering
\includegraphics[width=3.75cm]{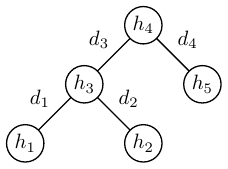}
\par(b)
\end{minipage}
\begin{minipage}{0.33\hsize}
\centering
\includegraphics[width=5.5cm]{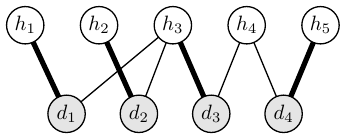}
\par(c)
\end{minipage}
\caption{(a) An example of $G_X$. 
(b) The tree rooted at $h_4$. 
(c) The desired matching.} 
\label{fig:tree}
\end{figure}

Define the simple directed graph ${\bf D} = (\mathcal{V},A)$ as follows. 
Define $\mathcal{V} := \mathcal{P}^{\ast} \cup \mathcal{Q} \cup \mathcal{R}$.
Then for each pair of distinct elements $X,Y \in \mathcal{V}$, 
there exists an arc from $X$ to $Y$ in $A$ if and only if 
the following conditions are satisfied. 
\begin{description}
\item[(P1)]
$Y \in \mathcal{P}^{\ast}$, 
$|E(d_Y)| = 2$, 
$\overline{h}_Y \in H_X$, and 
$(d_Y,h_Y) \succ_{d_Y} (d_Y,\overline{h}_Y)$. 
\item[(P2)]
If $X \in \mathcal{P}^{\ast}$, then 
$h_X = \overline{h}_Y$ and 
$(d_Y,h_X) \succ_{h_X} (d_X,h_X)$. 
\item[(P3)]
If $X \in \mathcal{Q}$, then 
$(d_Y,\overline{h}_Y) \succsim_{\overline{h}_Y} e$ for an edge $e \in L(\overline{h}_Y)$.
\item[(P4)] 
If $X \in \mathcal{R}$, then, for every 
edge $e = (d,\overline{h}_Y) \in R(\overline{h}_Y)$, 
the following conditions are satisfied.
\begin{description}
\item[i)]
If 
$e \sim_d \mu(d)$, then 
$(d_Y,\overline{h}_Y) \succsim_{\overline{h}_Y} e$. 
\item[ii)]
If 
$e \succ_d \mu(d)$, then 
$(d_Y,\overline{h}_Y) \succ_{\overline{h}_Y} e$. 
\end{description}
\end{description}
Notice that, in (P1), (R5) implies that 
$(d_Y,\overline{h}_Y) \notin R$. 
Furthermore,
in (P4), 
(R5) implies that 
$e \succsim_d \mu(d)$. 
For every element $X \in \mathcal{V}$, at most one arc of ${\bf D}$
enters $X$.

Define the subsets $\mathcal{V}^+, \mathcal{V}^- \subseteq \mathcal{V}$ as follows. 
For each element $X \in \mathcal{V}$, 
$X$ is contained in $\mathcal{V}^+$ if and only if 
(i) $X \in \mathcal{Q}$ and $H_X \subseteq H\setminus S$, or 
(ii) $X \in \mathcal{R}$.
For each element $X \in \mathcal{V}$, 
$X$ is contained in $\mathcal{V}^-$ if and only if 
$X \in \mathcal{P}^{\ast}$ and 
$h_X \in S$. 
Notice that, for any element $X \in \mathcal{V}^+$, 
there does not exist an arc of ${\bf D}$ entering $X$. 

Our algorithm for this special case is based on the following 
lemmas. We give the proofs of these lemmas in 
Sections~\ref{section:proof:lemma_1:theorem:degree_2} 
and \ref{section:proof:lemma_2:theorem:degree_2}. 

\begin{lemma} \label{lemma_1:theorem:degree_2} 
If, for every element $X \in \mathcal{V}^+$, 
there exists an element $Y \in \mathcal{V}^-$ such that 
there exists a directed path from $X$ to $Y$ 
in ${\bf D}$, then 
there exists a stable matching in $G$.
\end{lemma}

\begin{lemma} \label{lemma_2:theorem:degree_2} 
If there exists a stable matching in $G$,
then, for every element $X \in \mathcal{V}^+$, 
there exists an element $Y \in \mathcal{V}^-$ such that 
there exists a directed path from $X$ to $Y$ 
in ${\bf D}$. 
\end{lemma}

\begin{theorem} \label{theorem:degree_2} 
There exists a stable matching in $G$ 
if and only if, for every element $X \in \mathcal{V}^+$, 
there exists an element $Y \in \mathcal{V}^-$ such that 
there exists a directed path from $X$ to $Y$ 
in ${\bf D}$. 
\end{theorem}
\begin{proof}
This theorem immediately follows from 
Lemmas~\ref{lemma_1:theorem:degree_2} and 
\ref{lemma_2:theorem:degree_2}. 
\end{proof} 

Theorem~\ref{theorem:degree_2} implies that 
the special case in this section can be solved in polynomial time
by checking the existence of a desired directed path in ${\bf D}$ for each 
element in $\mathcal{V}^+$. 
Furthermore, the proof of Lemma~\ref{lemma_1:theorem:degree_2}
gives a polynomial-time algorithm for finding a
stable matching if exists. 

\subsection{Proof of Lemma~\ref{lemma_1:theorem:degree_2}}
\label{section:proof:lemma_1:theorem:degree_2}

In this subsection, we assume that,  
for every element $X \in \mathcal{V}^+$, 
there exists an element $Y_X \in \mathcal{V}^-$ such that 
there exists a directed 
path from $X$ to $Y_X$ 
in ${\bf D}$. 
For each element $X \in \mathcal{V}^+$, 
let $P_X$ be a directed 
path from $X$ to $Y_X$ in ${\bf D}$.
We assume that, 
for each element $X \in \mathcal{V}^+$, 
$P_X$ goes through vertices  
$X_1,X_2,\dots,X_{\ell(X)}$
of ${\bf D}$ in this order.
Notice that, for each element $X \in \mathcal{V}^+$, 
$X_{\ell(X)} = Y_X$ holds. 
Since
at most one arc of ${\bf D}$ enters 
$X$ for every element $X \in \mathcal{V}$,
the directed paths $P_X, P_Z$ are vertex-disjoint for every pair of 
distinct elements $X,Z \in \mathcal{V}^+$.  

For each element $X \in \mathcal{P}^{\ast}$, 
if there exist an element $Y \in \mathcal{V}^+$ 
and an integer $i \in [\ell(Y)]$ 
such that 
$Y_i = X$, then we say that 
\emph{$X$ is on $P_Y$}. 

\begin{figure}[t]
\begin{minipage}{0.49\hsize}
\centering
\includegraphics[width=7cm]{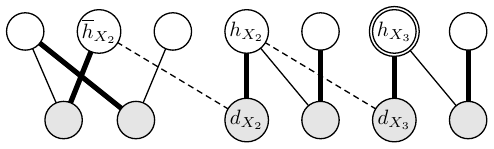}
\par(a)
\end{minipage}
\begin{minipage}{0.49\hsize}
\centering
\includegraphics[width=7cm]{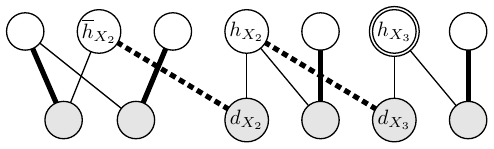}
\par(b)
\end{minipage}
\caption{(a) The real lines represent edges in $L$.
The broken lines represent edges in $E \setminus L$. 
The bold lines represent edges in $\mu$.  
(b) The bold lines represent edges in $\sigma$.}
\label{fig:replace}
\end{figure}

Define $\mu^{\bullet}$ as the matching obtained from $\mu$ 
in the following way.
For each element $X \in \mathcal{Q} \cap \mathcal{V}^+$, 
we replace 
the edges in $\mu$ contained in 
$G^{\ast}\langle X \rangle$
with the matching $\xi$ in 
Lemma~\ref{lemma:tree}
such that $\xi(\overline{h}_{X_2}) = \emptyset$ 
(see Figure~\ref{fig:replace}). 

Define 
$\mu^{\circ}$ as the matching obtained from $\mu^{\bullet}$
in the following way. 
For each element $X \in \mathcal{Q} \setminus \mathcal{V}^+$, 
we define ${\sf fr}_X$ as an arbitrary hospital in $H_X \cap S$, 
and then we replace 
the edges in $\mu^{\bullet}$ 
contained in 
$G^{\ast}\langle X \rangle$
with the matching $\xi$ in 
Lemma~\ref{lemma:tree}
such that $\xi({\sf fr}_X) = \emptyset$. 

Finally, we define the matching $\sigma$ as follows. 
Define $M_+,M_-$ by 
\begin{equation*}
\begin{split}
M_+ & := \bigcup_{X \in \mathcal{V}^+}\{(d_{X_2},\overline{h}_{X_2}),
(d_{X_3},\overline{h}_{X_3}),
\dots,(d_{X_{\ell(X)}},\overline{h}_{X_{\ell(X)}})\},\\ 
M_- & := \bigcup_{X \in \mathcal{V}^+}\{(d_{X_2},h_{X_2}),
(d_{X_3},h_{X_3}),\dots,(d_{X_{\ell(X)}},h_{X_{\ell(X)}})\}. 
\end{split} 
\end{equation*}
Define $\sigma := (\mu^{\circ} \setminus M_-) \cup M_+$
(see Figure~\ref{fig:replace}).  
It is not difficult to see that 
$\sigma$ is a matching in $G$. 
Thus, what remains is to prove that 
$\sigma$ is stable. 

In the rest of this subsection, let us fix an edge $e = (d,h) \in E \setminus \sigma$. 
Then we prove that $e$ does not block $\sigma$. 
In what follows, we divide the proof into the following three cases. 
\begin{description}
\item[Case~1.] 
There exists an element $X \in \mathcal{P}$
such that 
$e$ is an edge of 
$G\langle X \rangle$.
\item[Case~2.]
There exists an element $X \in \mathcal{Q}$
such that 
$e$ is an edge of 
$G\langle X \rangle$.
\item[Case~3.]
There does not exist an element $X \in \mathcal{P} \cup \mathcal{Q}$ 
such that 
$e$ is an edge of 
$G\langle X \rangle$.
\end{description}

\subsubsection{Case~1}

If $X \notin \mathcal{P}^{\ast}$, i.e.,  
$|L(d)|=2$ for every doctor $d \in D_X$, then
$e \in L$, $\sigma(d) \in L$, and $\sigma(h) \in L$. 
This implies that $e$ does not block $\sigma$. 
Thus, in what follows, 
we assume that $X \in \mathcal{P}^{\ast}$. 

We first consider the case where 
there does not exist an element $Y \in \mathcal{V}^+$ 
such that $X$ is on $P_Y$. 
In this case, $\sigma(d) = \mu(d) \in L$ and 
$\sigma(h) = \mu(h) \in L$.
Thus, if $e \in L$, then $e$ does not block $\sigma$.
This implies that we can assume that $e \notin L$. 
If $e \in R$ and $e$ blocks $\sigma$, then 
$e \in R \cap {\sf block}(L)$.
However, this contradicts (R4).
This implies that if $e \in R$, then 
$e$ does 
not block $\sigma$.
Thus, we assume that $e \notin R$. 
If $\mu(d) \succ_d e$, then 
$e$ does not block $\sigma$.
Thus, we assume that 
$e \succsim_d \mu(d)$.
In this case, since $e \notin R$,
$e \in {\rm Ch}_d(E \setminus R)$. 
Thus, since $e \notin L$ and $\mu(h) \in L$,
we have $\mu(h) \succ_h e$. 
This implies that 
$e$ does not block $\sigma$.

Next we consider the case where 
there exists an element $Y \in \mathcal{V}^+$ 
such that $X$ is on $P_Y$. 
In this case, we have $e \in L$. 
If $h \neq h_{X}$, then  
$\sigma(d) = \mu(d) \in L$, and $\sigma(h) = \mu(h) \in L$. 
Thus, $e$ does not block $\sigma$. 
Assume that $h = h_X$. 
If $X = Y_{\ell(Y)}$, 
then $h \in S$ and 
$\sigma(h) = \emptyset$.
Thus, $e$ does 
not block $\sigma$. 
Assume that 
$X \neq Y_{\ell(Y)}$.
If $d \neq d_X$, then $\sigma(d) \in L$ and $\sigma(h) \succsim_h \mu(h) \in L$. 
Thus, $e$ does not block $\sigma$.
If $d = d_X$ (i.e., $e = (d_{X},h_{X})$), then 
since $\sigma(h) \succ_{h} e$,
$e$ does not block $\sigma$. 

\subsubsection{Case~2} 

In this case, $e \in L$. 
Thus, if $\sigma(h) \in L$, then 
$\sigma(d) \sim_d e$ and $\sigma(h) \sim_h e$. 
If $\sigma(h) \notin L$ and $\sigma(h) \neq \emptyset$, then 
since $e \in L$, (P3) implies that 
$\sigma(d) \sim_d e$ and $\sigma(h) \succsim_h e$. 
If $\sigma(h) = \emptyset$, then 
$h \in S$. 
Thus, $e$ does not block $\sigma$. 

\subsubsection{Case~3} 

In this case, $e \notin L$. 
Otherwise, $e$ is an edge of 
$G\langle X \rangle$
for some element $X \in \mathcal{P} \cup \mathcal{Q}$.

\begin{claim} \label{claim_1:case_3:lemma_1:theorem:degree_2} 
$\sigma(d) = \mu(d)$. 
\end{claim}
\begin{proof}
If $L(d) = 0$, then 
$d \notin D_X$ holds for every element $X \in \mathcal{P} \cup \mathcal{Q}$. 
This and the definition of $\sigma$ implies that 
$\sigma(d) = \mu(d) = \emptyset$. 
Assume that 
$L(d) \neq \emptyset$. Then 
there exists an element 
$X \in \mathcal{P}^{\ast}$ such that 
$e = (d_X,\overline{h}_X)$.
Since $e \notin \sigma$, 
there does not exist an element $Y \in \mathcal{V}^+$ 
such that $X$ is on $P_Y$.
Thus, $\sigma(d) = \mu(d) = (d_X,h_X)$.
This completes the proof. 
\end{proof} 

If there exists an element $Z \in \mathcal{P} \cup \mathcal{Q}$ 
such that $h \in H_Z$ and $\sigma(h) = \emptyset$, 
then $h \in S$. 
In this case, $e$ does not block $\sigma$. 
Thus, it is sufficient to consider the following cases. 

\begin{description}
\item[(3a)] 
There exists an element $Z \in \mathcal{P}$ 
such that $h \in H_Z$ and $\sigma(h) \neq \emptyset$.
\item[(3b)]
There exists an element $Z \in \mathcal{Q}$ 
such that $h \in H_Z$ and $\sigma(h) \neq \emptyset$. 
\item[(3c)]
There exists an element $Z \in \mathcal{R}$ 
such that $h \in H_Z$.
\item[(3d)]
There does not exist an element $Z \in \mathcal{P} \cup \mathcal{Q} \cup \mathcal{R}$ 
such that $h \in H_Z$.
\end{description}

Notice that, in the cases (3c) and (3d), 
we have 
$L(h) = \emptyset$. 

\begin{claim} \label{claim_2:case_3:lemma_1:theorem:degree_2} 
In the cases {\rm (3a)} and {\rm (3b)},
if $e \succsim_d \mu(d)$ and $e \notin R$, then 
$e$ does not block $\sigma$. 
\end{claim}
\begin{proof}
Assume that $e \succsim_d \mu(d)$ and $e \notin R$. 
Lemma~\ref{lemma:property_R}
implies that $\mu(d) \neq \emptyset$. 
Since $e \succsim_d \mu(d)$ and $\mu(d) \in L$, 
$e \in {\rm Ch}_d(E \setminus R) \cap E(h)$. 
Thus, $L(h) \neq \emptyset$. 
Since $e \notin L$, 
$f \succ_h e$ for an edge $f \in L(h)$. 
This implies that since (P2) and (P3) imply that $\sigma(h) \succsim_h L$, 
$e$ does not 
block $\sigma$. 
\end{proof} 

\begin{claim} \label{claim_3:case_3:lemma_1:theorem:degree_2} 
In the cases {\rm (3c)} and {\rm (3d)},
if $e \succsim_d \mu(d)$, 
then $e \in R$.
\end{claim}
\begin{proof}
Assume that $e \succsim_d \mu(d)$ and $e \notin R$. 
Then Lemma~\ref{lemma:property_R}
implies that $\mu(d) \neq \emptyset$. 
Since $e \succsim_d \mu(d)$ and $\mu(d) \in L$, 
$e \in {\rm Ch}_d(E \setminus R) \cap E(h)$. 
However, this contradicts the fact $L(h) = \emptyset$. 
\end{proof} 

Claim~\ref{claim_1:case_3:lemma_1:theorem:degree_2}
implies that if $\mu(d) \succ_d e$, then  
$e$ does not block $\sigma$. 
Thus, in what follows, we assume that $e \succsim_d \mu(d) = \sigma(d)$. 
In this case, 
Claims~\ref{claim_2:case_3:lemma_1:theorem:degree_2}
and \ref{claim_3:case_3:lemma_1:theorem:degree_2} 
imply that we can assume that 
$e \in R$.

First, we consider the case (3a).
In this case, $L(h) \neq \emptyset$.
If $e$ blocks $\sigma$, then 
since (P2) implies that $\sigma(h) \succsim_h \mu(h)$, 
$e$ blocks $\mu$.
Thus, since $\mu \subseteq L$, 
$e \in R \cap {\sf block}(L)$.
However, this contradicts (R4). 
Thus, $e$ does not block $\sigma$.

Next, we consider the case (3b).  
In this case,
(P3) implies that $\sigma(h) \succsim_h L$.
Thus, if $e$ blocks $\sigma$, then 
since $L(h) \neq \emptyset$, 
$e \in R \cap {\sf block}(L)$.
This contradicts (R4). 
Thus, $e$ does not block $\sigma$. 

Third, we consider the case (3c).
In this case, Claim~\ref{claim_1:case_3:lemma_1:theorem:degree_2}
and (P4) imply that 
(i) $\sigma(h) \succsim_h e$ if $e \sim_d \sigma(d)$, and 
(ii) $\sigma(h) \succ_h e$ if $e \succ_d \sigma(d)$.
Thus, $e$ does not block $\sigma$.

Forth, 
we consider the case (3d).
The definition of $\sigma$ implies that 
$\sigma(h) = \mu(h) = \emptyset$. 
Since $e \in R(h)$ and $Z \notin \mathcal{R}$, 
we have $h \in S$. 
Thus, since $\sigma(h) = \emptyset$, 
$e$ does not block $\sigma$.

\subsection{Proof of Lemma~\ref{lemma_2:theorem:degree_2}}
\label{section:proof:lemma_2:theorem:degree_2}

Assume that there exists a stable matching $\sigma$ in $G$. 

Let $X$ be an element in $\mathcal{V}^+$.
Define the subset $A_X \subseteq A$ as follows.
Let $a$ be an arc in $A$ from $Z^{\prime} \in \mathcal{V}$ to $Z \in \mathcal{V}$.
Then $a \in A_X$ if and only if 
$(d_Z,\overline{h}_Z) \in \sigma$. 
Define the subgraph ${\bf D}_X$ of ${\bf D}$ by 
${\bf D}_X := (\mathcal{V},A_X)$.
Furthermore, let $\mathcal{V}_X$ be the set of elements in $\mathcal{V}$
that are reachable from $X$ in ${\bf D}_X$ via a directed path. 
If $\mathcal{V}_X$ contains an element $Y \in \mathcal{V}^-$,  
then proof is done since 
there exists a directed path from $X$ to $Y$ in ${\bf D}$. 
Thus, we consider the case where 
$\mathcal{V}_X$ does not contain an element in $\mathcal{V}^-$.
Recall that 
at most one arc of ${\bf D}$ enters 
$Z$ for every element $Z \in \mathcal{V}$.
In addition, for any element $Z \in \mathcal{V}^+$, 
there does not exist an arc of ${\bf D}$ entering 
$Z$. 
This implies that 
${\bf D}_X$ is acyclic. 
Thus, there exists an element $T \in \mathcal{V}_X$ such that 
any arc in $A_X$ does not leave $T$.  
In what follows, we divide the proof into the following two cases. 
\begin{description}
\item[Case~1.] 
$T \in \mathcal{V}^+$, i.e., $T = X$. 
\item[Case~2.]
$T \in \mathcal{P}^{\ast}$. 
\end{description}

\subsubsection{Case~1} 

We first consider 
the case where $T \in \mathcal{Q}$.
Since $|D_T| < |H_T|$, there exists 
a hospital $h \in H_{T}$ such that 
$\sigma(h) \notin L$. 
Let $(d,h)$ be an arbitrary edge in $L(h)$. 
Then since (R1) implies that $\sigma \subseteq E \setminus R$, 
we have 
$(d,h) \succsim_d \sigma(d)$. 
Thus, if $\sigma(h) = \emptyset$, then 
$(d,h)$ blocks $\sigma$. 
However, this contradicts the fact that 
$\sigma$ is stable. 
Thus, we consider the case where 
$\sigma(h) \neq \emptyset$. 
Assume that $\sigma(h) = (d^{\prime},h)$.
Then since $\sigma$ is stable,    
$(d^{\prime},h) \succsim_h (d,h)$.
If $d^{\prime} \notin D_Z$ holds for every element $Z \in \mathcal{P}^{\ast}$, 
then since $L(d^{\prime}) = \emptyset$, Lemma~\ref{lemma:property_R} 
implies that $(d^{\prime},h) \in R$. 
However, this contradicts the fact that $\sigma \subseteq E \setminus \sigma$. 
Thus, there exists an element $Z \in \mathcal{P}^{\ast}$ 
such that $d_Z = d^{\prime}$. 
Assume that 
$(d^{\prime},h) \succsim_{d^{\prime}} (d^{\prime},h_Z)$.
Since $(d^{\prime},h) = \sigma(h) \in E\setminus R$, 
we have $(d^{\prime},h) \in {\rm Ch}_{d^{\prime}}(E \setminus R)$.
Since $\sigma(h) \notin L$, 
$e \succ_h \sigma(h)$ 
for an edge $e \in L(h)$.
However, since $L$ is flat and $(d,h) \in L$, 
this contradicts the fact that 
$(d^{\prime},h) \succsim_h (d,h)$.
Thus, we have $(d^{\prime},h_Z) \succ_{d^{\prime}} (d^{\prime},h)$. 
This implies that there exists an arc from $X$ to $Z$ in ${\bf D}$, and 
$A_X$ contains this arc. 
However, this contradicts the fact that 
any arc in $A_X$ does not leave $T$. 

Next, we consider the case where 
$T \in \mathcal{R}$. 
Assume that $T = \{h\}$. 
Recall that since $T \in \mathcal{R}$,
$h \in H \setminus S$ and $R(h) \neq \emptyset$.
If $\sigma(h) = \emptyset$, then
every edge $e = (d,h) \in R$ blocks $\sigma$  
since (R1) and (R5) imply that $e \succsim_d \sigma(d)$. 
Thus, we consider the case where $\sigma(h) \neq \emptyset$. 
Since $T \in \mathcal{R}$, 
$\sigma(h) \notin L$.
Assume that $\sigma(h) = (d^{\prime},h)$.
If $d^{\prime} \notin D_Z$ holds for every element $Z \in \mathcal{P}^{\ast}$, 
then since $L(d^{\prime}) = \emptyset$, Lemma~\ref{lemma:property_R} 
implies that $(d^{\prime},h) \in R$. 
However, this contradicts the fact that $\sigma \subseteq E \setminus \sigma$. 
Thus, there exists an element $Z \in \mathcal{P}^{\ast}$ 
such that $d_Z = d^{\prime}$. 
Assume that 
$(d^{\prime},h) \succsim_{d^{\prime}} (d^{\prime},h_Z)$.
Then 
since $(d^{\prime},h) = \sigma(h) \in E\setminus R$, 
we have $(d^{\prime},h) \in {\rm Ch}_{d^{\prime}}(E \setminus R)$.
However, this contradicts the fact that $L(h) = \emptyset$. 
Thus, $(d^{\prime},h_Z) \succ_{d^{\prime}} (d^{\prime},h)$. 
In order to prove this case, 
we prove that, for every edge $e \in R(h)$, 
the following conditions (i) and (ii) are satisfied. 
\begin{description}
\item[(i)]
If $e \sim_d \mu(d)$, then $(d^{\prime},h) \succsim_h e$.
\item[(ii)]
If $e \succ_d \mu(d)$, then $(d^{\prime},h) \succ_h e$. 
\end{description} 
If we can prove this, then 
there exists an arc from $X$ to $Z$ in ${\bf D}$, and 
$A_X$ contains this arc. 
This contradicts the fact that 
any arc in $A_X$ does not leave $T$.

In order to prove the conditions (i) and (ii) are satisfied, we first prove that 
$\mu(d) \succsim_d \sigma(d)$. 
\begin{itemize}
\item
If $\mu(d) = \emptyset$, 
then 
(R3) implies that 
$E(d) \setminus R = \emptyset$.
Thus, since $\sigma \subseteq E \setminus R$, 
$\sigma(d) = \emptyset$.
This implies that, in this case, $\mu(d) \sim_d \sigma(d)$
\item
If $\mu(d) \neq \emptyset$, then 
since (R1) implies that $\sigma \subseteq E \setminus R$, 
(R2) implies that $\mu(d) \succsim_d \sigma(d)$. 
\end{itemize}
We continues the proof of the conditions (i) and (ii). 
Let $e = (d,h)$ be an edge in $R(h)$. 
\begin{itemize}
\item
Assume that $e \sim_d\mu(d)$ and 
$e \succ_h (d^{\prime},h)$ (i.e., 
the condition (i) is not satisfied). 
Then 
since $\mu(d) \succsim_d \sigma(d)$, we have 
$e \succsim_d \sigma(d)$. 
Thus, since $e \succ_h (d^{\prime},h) = \sigma(h)$, 
$e$ blocks $\sigma$. 
However, this contradicts the fact that $\sigma$ is stable. 
\item
Assume that $e \succ_d \mu(d)$ and 
$e \succsim_h (d^{\prime},h)$ (i.e., the condition (ii) is not satisfied.) 
Then 
since $\mu(d) \succsim_d \sigma(d)$, we have 
$e \succ_d \sigma(d)$. 
Thus, since $e \succsim_h (d^{\prime},h) = \sigma(h)$, 
$e$ blocks $\sigma$. 
However, this contradicts the fact that $\sigma$ is stable. 
\end{itemize}

\subsubsection{Case~2} 

Since $T \neq X$, some arc of ${\bf D}_X$ enters $T$. 
This implies that $\sigma(d_T) = (d_T,\overline{h}_T)$ and  
$(d_T,h_T) \succ_{d_T} (d_T,\overline{h}_T)$.
If $\sigma(h_T) = \emptyset$, then 
since $T \notin \mathcal{V}^-$ (i.e., $h_T \notin S$), 
$(d_T,h_T)$ blocks $\sigma$. 
Thus, we consider the case where $\sigma(h_T) \neq \emptyset$. 
Assume that $\sigma(h_T) = (d,h_T)$.
Then since $(d_T,h_T) \succ_{d_T} \sigma(d_T)$ and 
$\sigma$ is stable, 
$(d,h_T) \succ_{h_T} (d_T,h_T)$. 
If $d \notin D_Z$ holds for every element $Z \in \mathcal{P}^{\ast}$, 
then since $L(d) = \emptyset$, Lemma~\ref{lemma:property_R} 
implies that $(d,h_T) \in R$. 
This contradicts the fact that $\sigma \subseteq E \setminus \sigma$. 
Thus, there exists an element $Z \in \mathcal{P}^{\ast}$ 
such that $d_Z = d$. 
Assume that 
$(d,h_T) \succsim_{d} (d,h_Z)$.
Since $(d,h_Z) \in L$ and 
$(d,h_T) \in E\setminus R$, 
we have $(d,h_T) \in {\rm Ch}_{d}(E \setminus R)$.
Since $T \neq Z$, we have 
$\sigma(h_T) \notin L$.
Thus, 
since $(d_T,h_T) \in L$, 
$(d_T,h_T) \succ_h (d,h_T)$. 
However, this contradicts the fact that 
$(d,h_T) \succ_h (d_T,h_T)$. 
Thus, 
we have $(d,h_Z) \succ_{d} (d,h_T)$.
This implies that 
there exists an arc from $X$ to $Z$ in ${\bf D}$, and 
$A_X$ contains this arc. 
However, 
this contradicts the fact that 
any arc in $A_X$ does not leave $T$.

\section{Proof of Lemma~\ref{lemma:pre-process}} 
\label{section:pre-process} 

In this section, we prove Lemma~\ref{lemma:pre-process}.
Although the setting of this paper and the setting in \cite{Irving94,Manlove99} 
are different, 
this lemma can be proved by using an idea similar to the 
idea used in \cite{Irving94,Manlove99}. 

For each subset $F \subseteq E$, 
we define the function $\rho_F$ on $2^{D[F]}$ by 
\begin{equation*}
\rho_F(X) := |\Gamma_F(X)| - |X|
\end{equation*}
for each subset $X \subseteq D[F]$. 
For every subset $F \subseteq E$,
it is not difficult to see that 
\begin{equation*}
\rho_F(X) + \rho_F(Y) \ge \rho_F(X \cup Y) + \rho_F(X \cap Y)
\end{equation*}
for every pair of 
subsets $X,Y \subseteq D[F]$ (i.e., $\rho_F$ is submodular).
For every subset $F \subseteq E$, 
it is known that 
there exists the unique (inclusion-wise)  
minimal minimizer of $\rho_F$, and we can find it in polynomial 
time (see, e.g., \cite[Note~10.12]{Murota03}). 
Furthermore, it is well known that, for every 
subset $F \subseteq E$, there exists a matching 
$\mu$ in $G$ such that 
$\mu \subseteq F$ and $|\mu| = |D[F]|$ if and only if, 
for every subset $X \subseteq D[F]$, 
$|X| \le |\Gamma_F(X)|$ holds~\cite{Hall35} (see also \cite[Theorem~16.7]{Schrijver02}). 

We are now ready to propose our algorithm for
proving Lemma~\ref{lemma:pre-process} 
(see Algorithm~\ref{alg:pre-process}). 

\begin{algorithm}[h]
Set $t := 0$. Define ${\sf R}_0 := \emptyset$.\\
\Do{${\sf R}_t \neq {\sf P}_{t,i_t}$}
{
  Set $t := t + 1$ and $i := 0$. Define ${\sf P}_{t,0} := {\sf R}_{t-1}$.\\
  \Do{${\sf P}_{t,i} \neq {\sf P}_{t,i-1}$}
  {
    Set $i := i+1$.\\
    Define $L_{t,i} := {\rm Ch}(E \setminus {\sf P}_{t,i-1})$.\\ 
    Define ${\sf Q}_{t,i} := {\sf P}_{t,i-1} \cup 
    ({\rm Ch}_D(E \setminus {\sf P}_{t,i-1}) \setminus L_{t,i})$.\\ 
    \uIf{${\sf P}_{t,i-1} \subsetneq {\sf Q}_{t,i}$}
    {
        Define ${\sf P}_{t,i} := {\sf Q}_{t,i}$.
    }    
    \Else
    {
        Define the bipartite graph $G_{t,i}$ by $G_{t,i} := (D, H; L_{t,i})$.\\
        Find a maximum-size matching $\mu_{t,i}$ in $G_{t,i}$.\\
        \uIf{$|\mu_{t,i}| < |D[E \setminus {\sf P}_{t,i-1}]|$}
        {
          Find the minimal minimizer $N_{t,i}$ of $\rho_{L_{t,i}}$.\\
          Define ${\sf P}_{t,i} := {\sf Q}_{t,i} \cup L_{t,i}(N_{t,i})$. 
        }
        \Else
        {
          Define ${\sf P}_{t,i} := {\sf Q}_{t,i}$.  
        }
    }
   
   }  
   Define $i_t := i$.\\
   \uIf{${\sf P}_{t,i_t} \cap {\sf block}({\rm Ch}(E \setminus {\sf P}_{t,i_t})) \neq \emptyset$}
   {
       Define $b_t = (d_t,h_t)$ as an edge in ${\sf P}_{t,i_t} 
       \cap {\sf block}({\rm Ch}(E \setminus {\sf P}_{t,i_t}))$.\\
       Define ${\sf R}_t := {\sf P}_{t,i_t} \cup ({\rm Ch}(E \setminus {\sf P}_{t,i_t}) \cap E(h_t))$. 
   }
   \Else
   {
       Define ${\sf R}_t := {\sf P}_{t,i_t}$. 
   }
}
Define $k := t$. Output ${\sf R}_k$ and $\mu_{k,i_k}$, and halt.
\caption{Algorithm for proving Lemma~\ref{lemma:pre-process}}
\label{alg:pre-process}
\end{algorithm}

\begin{lemma} \label{lemma:pre-process_iteration}
In the course of Algorithm~\ref{alg:pre-process}, 
the following statements hold. 
\begin{description}
\item[(i)]
In each iteration of Steps~2 to 28, 
the number of iterations of Steps~4 to 20
is at most $|E|$.
\item[(ii)]
The number of iterations of Steps~2 to 28 
is at most $|E|$.
\end{description}
\end{lemma}
\begin{proof}
The statement (i) follows from the fact that ${\sf P}_{t,i-1} \subseteq {\sf P}_{t,i}$.
Furthermore, 
the statement (ii) follows from the fact that ${\sf R}_{t-1} \subseteq {\sf R}_{t}$.
This completes the proof. 
\end{proof} 

Since Step~11 of Algorithm~\ref{alg:pre-process} 
can be done in polynomial time 
by using, e.g., the algorithm in \cite{HopcroftK73}, 
Lemma~\ref{lemma:pre-process_iteration} implies that 
Algorithm~\ref{alg:pre-process} is a polynomial-time 
algorithm. 

\begin{lemma} \label{lemma:ch_d}
In the course of Algorithm~\ref{alg:pre-process}, 
if 
${\sf P}_{t,i-1} = {\sf Q}_{t,i}$, then 
$|D[L_{t,i}]| = |D[E \setminus {\sf P}_{t,i-1}]|$. 
\end{lemma}
\begin{proof}
If 
${\sf P}_{t,i-1} = {\sf Q}_{t,i}$, then
$L_{t,i} = {\rm Ch}_D(E \setminus {\sf P}_{t,i-1})$. 
Thus, 
\begin{equation*}
D[E \setminus {\sf P}_{t,i-1}] = D[{\rm Ch}_D(E \setminus {\sf P}_{t,i-1})] = D[L_{t,i}].
\end{equation*}
This completes the proof.
\end{proof}

In what follows, we define $\mu := \mu_{k,i_k}$ and 
$R := {\sf R}_k$. 

\begin{lemma} \label{lemma:R1}
There does not exist a stable matching in $G$ 
containing an edge in $R$.
\end{lemma} 
\begin{proof}
An edge $e \in R$ is called a \emph{bad edge} if 
there exists a stable matching in $G$ containing 
$e$. 
If we can prove that there does not exist a bad edge 
in $R$, then the proof is done. 
Assume that there exists a bad edge in $R$. 
Define $\Delta$ as the set of integers $i \in [k]$ such that 
${\sf R}_i \setminus {\sf R}_{i-1}$ contains 
a bad edge in $R$. 
Let $z$ be the minimum integer in $\Delta$. 

First, we consider the case where 
${\sf P}_{z,i_z} \setminus {\sf R}_{z-1}$ contains a bad edge in $R$. 
We define $j$ as the minimum integer in $[i_z]$ such that 
${\sf P}_{z,j} \setminus {\sf P}_{z,j-1}$ contains a bad edge in $R$.
Notice that, in this case, 
${\sf P}_{z,j-1}$ does not contain a bad edge in $R$. 

Assume that ${\sf Q}_{z,j} \setminus {\sf P}_{z,j-1}$ contains 
a bad edge in $R$. 
Let $f = (p,q)$ be a bad edge contained in ${\sf Q}_{z,j} \setminus {\sf P}_{z,j-1}$.
Then there exists a stable matching $\sigma$ in $G$ such that 
$f \in \sigma$. 
Since $f \in {\sf Q}_{z,j} \setminus {\sf P}_{z,j-1}$, 
there exists an edge $e = (d,q) \in {\rm Ch}_D(E \setminus {\sf P}_{z,j-1})$
such that $e \succ_q f$. 
If $\sigma(d) \succ_d e$, then we have
$\sigma(d) \in {\sf P}_{z,j-1}$. 
However, $\sigma(d)$ is a bad edge in $R$, and 
this contradicts the fact that 
${\sf P}_{z,j-1}$ does not contain a bad edge in $R$. 
Thus, $e \succsim_d \sigma(d)$, and 
$e$ blocks $\sigma$. 
This contradicts the fact that $\sigma$ is stable. 

Assume that ${\sf Q}_{z,j} \setminus {\sf P}_{z,j-1}$ does not contain 
a bad edge in $R$. 
In this case, ${\sf P}_{z,j} \setminus {\sf Q}_{z,j}$ contains a bad edge in $R$. 
Let $f = (p,q)$ be a bad edge contained in ${\sf P}_{z,j} \setminus {\sf Q}_{z,j}$.
Then there exists a stable matching $\sigma$ in $G$ such that 
$f \in \sigma$. 
In this case, $f \in L_{z,j}(N_{z,j})$. 
Define $L := L_{z,j}$ and 
$N := N_{z,j}$. 
Notice that Lemma~\ref{lemma:ch_d}
implies that 
$|D[L]| = |D[E \setminus {\sf P}_{z,j-1}]|$. 
Thus, since $|\mu_{z,j}| < |D[E \setminus {\sf P}_{z,j-1}]|$, 
we have $|\mu_{z,j}| < |D[L]|$.
This implies that 
$|N| > |\Gamma_L(N)|$. 
Furthermore, $L = {\rm Ch}_D(E \setminus {\sf P}_{z,j-1})$. 

\begin{claim} \label{claim:R1_1}
There exists a doctor $d \in N$ satisfying the following 
conditions.
\begin{itemize}
\item
$\sigma(d) \notin L$.
\item
There exists a hospital $h \in \Gamma_L(d)$ such that 
$\sigma(h) \in L$.
\end{itemize} 
\end{claim}
\begin{proof}
Define $T$ as the set of doctors $d \in N$ 
such that $\sigma(d) \in L$.
Notice that the existence of $f$ implies that 
$T \neq \emptyset$. 
Then $|\Gamma_L(T)| \ge |T|$. 
Notice that $T \subsetneq N$. 
In order to prove this claim, 
it is sufficient to 
prove that 
a doctor $d \in N \setminus T$ such that 
$\Gamma_{L}(d) \cap \Gamma_{L}(T) \neq \emptyset$. 
If $\Gamma_{L}(d) \cap \Gamma_{L}(T) = \emptyset$ for every doctor 
$d \in N \setminus T$, 
then $\Gamma_L(N \setminus T) \subseteq \Gamma_L(N) \setminus \Gamma_{L}(T)$. 
Thus, 
\begin{equation} \label{eq1:claim:R1_1}
\begin{split}
& |\Gamma_L(N \setminus T)| - |N \setminus T|
\le 
|\Gamma_L(N) \setminus \Gamma_{L}(T)| - |N| + |T|\\
& =   
|\Gamma_L(N)| - |\Gamma_{L}(T)| - |N| + |T|
\le 
|\Gamma_L(N)| - |N|.
\end{split}
\end{equation}
Since $T \neq \emptyset$, $N\setminus T \subsetneq N$. 
Thus, \eqref{eq1:claim:R1_1} contradicts the fact that 
$N$ is the minimal minimizer of $\rho_{L}$. 
This completes the proof. 
\end{proof} 

Let $d$ be a doctor in $N$ satisfying the conditions in Claim~\ref{claim:R1_1}. 
Furthermore, let $h$ be 
a hospital in
$\Gamma_{L}(d)$ such that 
$\sigma(h) \in L$. 
Since 
$(d,h), \sigma(h) \in L$, 
we have $(d,h) \sim_h \sigma(h)$.  
Furthermore, since 
${\sf P}_{z,j-1}$ does not contain a bad edge in $R$ and $\sigma$ is stable, 
we have $\sigma(d) \in E \setminus {\sf P}_{z,j-1}$. 
Thus, since $\sigma(d) \notin L$ and $(d,h) \in L$, 
we have $(d,h) \succ_d \sigma(d)$. 
(Recall that $L = {\rm Ch}_D(E \setminus {\sf P}_{z,j-1})$.)
However, this implies that 
$(d,h)$ blocks $\sigma$.
This contradicts the fact that $\sigma$ is stable. 

Next, we consider the case where 
${\sf P}_{z,i_z} \setminus {\sf R}_{z-1}$ does not contain a bad edge in $R$. 
In this case, 
${\sf P}_{z,i_z}$ does not contain a bad edge in $R$. 
Define $K := {\rm Ch}(E \setminus {\sf P}_{z,i_z})$. 
Then $K(h_z)$ 
contains a bad edge in $R$.
Let $f$ be a bad edge in $K(h_z)$.
Then there exists a stable matching $\sigma$ in $G$ such that 
$f \in \sigma$.
Since ${\sf P}_{z,i_z-1} = {\sf P}_{z,i_z}$, we have 
\begin{equation*}
|\mu_{z,i_z}| = |D[E \setminus {\sf P}_{z,i_z-1}]| = |D[E \setminus {\sf P}_{z,i_z}]|.
\end{equation*} 
This implies that $\mu_{z,i_z}(d) \neq \emptyset$ holds 
for every doctor $d \in D[E \setminus {\sf P}_{z,i_z}]$. 
Thus, since $\mu_{z,i_z} \subseteq K$, 
we have $K(d) \neq \emptyset$
for every doctor $d \in D[E \setminus {\sf P}_{z,i_z}]$. 
Since ${\sf P}_{z,i_z}$ does not contain a bad edge in $R$, 
$\sigma \subseteq E \setminus {\sf P}_{z,i_z}$. 
Thus, if $d \in D[E \setminus {\sf P}_{z,i_z}]$, then 
$e \succsim_d \sigma(d)$ for an edge $e \in K(d)$. 
Furthermore, if $d \notin D[E \setminus {\sf P}_{z,i_z}]$ (i.e., 
$E(d) \subseteq {\sf P}_{z,i_z}$), then 
$\sigma(d) = \emptyset$. 
Since $b_z \in {\sf P}_{z,i_z}$, 
we have $b_z \notin \sigma$. 
Since $b_z \in {\sf block}(K)$, 
we have 
$b_z \succsim_{d_z} K$.
\begin{itemize}
\item
Assume that $b_z \succ_{d_z} K$. 
Then the above observation implies that 
$b_z \succ_{d_z} \sigma(d_z)$.
Furthermore, since
$b_z \in {\sf block}(K)$, 
$b_z \succsim_{h_z} K$.
Thus, since $f \in K$, 
$b_z \succsim_{h_z} f$. 
This implies that 
$b_z$ blocks $\sigma$.
However, this contradicts the fact that 
$\sigma$ is stable. 
\item
Assume that $b_z \sim_{d_z} K$. 
Then the above observation implies that 
$b_z \succsim_{d_z} \sigma(d_z)$.
Furthermore, since
$b_z \in {\sf block}(K)$, 
$b_z \succ_{h_z} K$.
Thus, since $f \in K$, 
$b_z \succ_{h_z} f$. 
This implies that 
$b_z$ blocks $\sigma$.
However, this contradicts the fact that 
$\sigma$ is stable. 
\end{itemize}

Thus, in any case, there does not exist a bad edge in $R$.
This completes the proof. 
\end{proof} 

\begin{lemma} \label{lemma:R}
The following statements hold.
\begin{description}
\item[(i)]
$\mu \subseteq {\rm Ch}(E \setminus R)$.
\item[(ii)]
$\mu(d) \neq \emptyset$ holds for every doctor $d \in D[E \setminus R]$.
\item[(iii)]
$R \cap {\sf block}({\rm Ch}(E \setminus R)) = \emptyset$. 
\item[(iv)]
For every doctor $d \in D$ and every pair 
of edges $e \in R(d)$ and $f \in E(d) \setminus R$,  
$e \succsim_d f$. 
\end{description}
\end{lemma}
\begin{proof} 
Notice that since Algorithm~\ref{alg:pre-process}
halts 
when $t = k$, we have 
\begin{equation} \label{eq_1:lemma:R}
{\sf P}_{k,i_k-1} = {\sf P}_{k,i_k} = {\sf R}_k.
\end{equation}
The statement (i) follows from the fact that 
$\mu \subseteq {\rm Ch}(E \setminus {\sf P}_{k,i_k-1})$
and \eqref{eq_1:lemma:R}. 
The statement (ii) follows from the fact that 
$|\mu_{k,i_k}| \ge |D[E\setminus {\sf P}_{k,i_k-1}]|$
and \eqref{eq_1:lemma:R}. 
The statement (iii) follows from the fact that
${\sf P}_{k,i_k} \cap {\sf block}({\rm Ch}(E \setminus {\sf P}_{k,i_k})) = \emptyset$
and \eqref{eq_1:lemma:R}. 
What remains is to prove the statement (iv).

Let $d$ be a doctor in $D$. 
For each subset $X \subseteq E(d)$, we consider 
the following statement.
\begin{equation} \label{eq_2:lemma:R}
\mbox{For every pair of edges $e \in X$ and
$f \in E(d) \setminus X$, we have 
$e \succsim_d f$.}
\end{equation} 
Our goal is to prove that 
\eqref{eq_2:lemma:R} holds when $X = R(d)$. 

Let $t$ be a positive integer such that $t \le k$.
Assume that
\eqref{eq_2:lemma:R} holds when $X = {\sf R}_{t-1}(d)$. 
Then we prove that 
\eqref{eq_2:lemma:R} holds when $X = {\sf R}_{t}(d)$. 
This completes the proof. 

Let $i$ be a positive integer such that $i \le i_t$.
Assume that 
\eqref{eq_2:lemma:R} holds when $X = {\sf P}_{t,i-1}(d)$.
If ${\sf P}_{t,i}(d) = {\sf P}_{t,i-1}(d)$, then 
\eqref{eq_2:lemma:R} holds when $X = {\sf P}_{t,i}(d)$.
Otherwise, since 
\begin{equation*}
{\sf P}_{t,i} \setminus {\sf P}_{t,i-1}
\subseteq {\rm Ch}_D(E \setminus {\sf P}_{t,i-1}),
\end{equation*}
\eqref{eq_2:lemma:R} holds when $X = {\sf P}_{t,i}(d)$.
This implies that 
\eqref{eq_2:lemma:R} holds when $X = {\sf P}_{t,i_t}(d)$. 

If ${\sf R}_t = {\sf P}_{t,i_t}$, then 
\eqref{eq_2:lemma:R} holds when $X = {\sf R}_{t}(d)$.
Otherwise, 
since 
\begin{equation*}
{\sf R}_{t} \setminus {\sf P}_{t,i_t} \subseteq 
{\rm Ch}_D(E \setminus {\sf P}_{t,i_t}),
\end{equation*}
\eqref{eq_2:lemma:R} holds when $X = {\sf R}_{t}(d)$.
This completes the proof of the statement (iv). 
\end{proof} 

Lemmas~\ref{lemma:R1} and \ref{lemma:R} imply that 
$R$ and $\mu$ satisfy 
the conditions (R1) to (R5).
This completes the proof.

\end{document}